%% file: sample.tex
\documentclass[11pt]{article}

\usepackage{dsfont}
\usepackage{amsmath}
\usepackage{amssymb}
\usepackage{graphicx}
\usepackage{url}
\usepackage{nicefrac}
\usepackage{hyperref}
\usepackage[letterpaper,margin=1in]{geometry}
\usepackage{multirow}
\usepackage{tikz} 
\usepackage{subcaption}
\usepackage{diagbox}
\usepackage{xspace}
\usepackage{amsmath, amsthm}
\usepackage{parskip}
\usetikzlibrary{arrows,decorations.markings}

\usepackage{cleveref}
\usepackage{algorithm}
\usepackage[noend]{algpseudocode}

\usepackage{color}

\definecolor{blueblack}{rgb}{0,0,.7}

\newcounter{sideremark}

\definecolor{Darkblue}{rgb}{0,0,0.4}
\definecolor{Brown}{cmyk}{0,0.61,1.,0.60}
\definecolor{Purple}{cmyk}{0.45,0.86,0,0}
\definecolor{brickred}{rgb}{0.8, 0.25, 0.33}

\newtheorem{theorem}{Theorem}[section]
\newtheorem{lemma}[theorem]{Lemma}
\newtheorem{definition}[theorem]{Definition}

\newtheorem{remark}[theorem]{Remark}
\newtheorem{quest}[theorem]{Question}
\newtheorem{cor}[theorem]{Corollary}

\newtheorem{claim}[theorem]{Claim}

\newcommand{\R}{\mathds{R}}

\newcommand{\poly}{\text{poly}}

\makeatletter
\newcommand{\newreptheorem}[2]{%
\newenvironment{rep#1}[1]{%
 \def\rep@title{\emph{\textbf{#2} \ref{##1}}}%
 \begin{rep@theorem}}%
 {\end{rep@theorem}}}
\makeatother
\newreptheorem{theorem}{Theorem}
\newreptheorem{lemma}{Lemma}
\newreptheorem{proposition}{Proposition}
\usepackage{thmtools,thm-restate}

\newcommand{\initOneLiners}{%
    \setlength{\itemsep}{0pt}
    \setlength{\parsep }{0pt}
    \setlength{\topsep }{0pt}
}

\newcommand{\eps}{\varepsilon}
\newcommand{\opt}{\text{OPT}}

\newcommand{\dist}{\text{dist}}

\newcommand{\cost}{\textsc{cost}}

\newcommand{\calC}{\mathcal{C}}

\newcommand{\calS}{\mathcal{S}}
\newcommand{\calN}{\mathcal{N}}

\newcommand{\calR}{\mathcal{R}}

\newcommand{\coreset}{\Omega}
\newcommand{\offset}{F}

\newcommand{\constantApprox}{\mathcal{A}}
\newcommand{\seeded}{\mathcal{G}}
\newcommand{\calB}{\mathcal{B}}
\newcommand{\inner}{R_I}
\newcommand{\out}{R_O}
\newcommand{\main}{R_M}
\newcommand{\polylog}{\text{polylog}}
\newcommand{\costP}{\textsc{P-cost}}
\newcommand{\conv}{\text{conv}}
\newcommand{\diam}{\text{diam}}
\newcommand{\cand}{\mathbb C}

\newcommand{\lpar}{\left(}
\newcommand{\rpar}{\right)}
\newcommand{\lbra}{\left\{}
\newcommand{\rbra}{\right\}}

\newcommand{\zext}[2]{\begin{pmatrix} #1 \\ #2\end{pmatrix}}

\DeclareMathOperator*{\argmin}{arg\,min}

\usepackage[colorinlistoftodos,textsize=tiny,textwidth=2cm,color=red!25!white,obeyFinal]{todonotes}

\def\DEBUG{true}

\ifdefined\DEBUG
  \def\rem#1{{\marginpar{\raggedright\scriptsize #1}}}

  \newcommand{\vir}[1]{\rem{\textcolor{Purple}{$\bullet$ #1}}}

  \newcommand{\dar}[1]{\rem{\textcolor{orange}{$\bullet$ #1}}}

\else

  \newcommand{\vir}[1]{}
  
  \newcommand{\dar}[1]{}
\fi

\usepackage{wrapfig}
\newcommand{\erclogowrapped}[1]{%
\setlength\intextsep{0pt}%
\begin{wrapfigure}[3]{r}{#1*\real{1.1}}%
\includegraphics[width=#1]{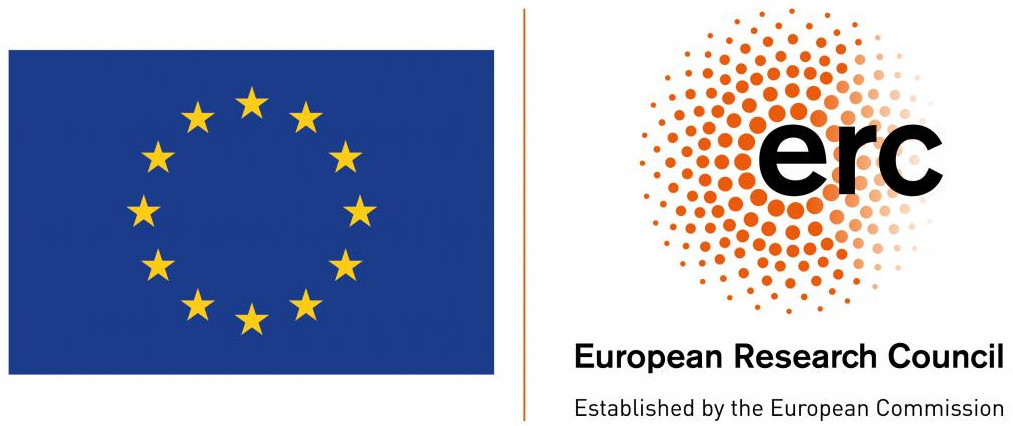}%
\end{wrapfigure}%
}

\begin{document}

        \title{Deterministic Clustering in High Dimensional Spaces: Sketches and Approximation}

\author{Vincent Cohen-Addad\thanks{Google Research, France} \and 
 David Saulpic\thanks{IST Austria. This project has received funding from the European Union’s Horizon 2020 research and innovation
programme under the Marie Skłodowska-Curie grant agreement No 101034413}
 \and
 Chris Schwiegelshohn \thanks{Aarhus University} \and
}
\date{}

\maketitle

\begin{abstract}
In all state-of-the-art sketching and coreset techniques for clustering, as well as in the best known fixed-parameter tractable approximation algorithms, randomness plays a key role.
For the classic $k$-median and $k$-means problems, there are no known deterministic  dimensionality reduction procedure or coreset construction that avoid an exponential dependency on the input dimension $d$, the precision parameter $\varepsilon^{-1}$ or $k$. Furthermore, there is no coreset construction that succeeds with probability $1-1/n$ and whose size does not depend on the number of input points, $n$.
This has led researchers in the area to ask what is the power of randomness for clustering sketches [Feldman WIREs Data Mining Knowl. Discov'20].

Similarly, the best approximation ratio achievable deterministically without a complexity exponential in the dimension are $1+\sqrt{2}$ for $k$-median [Cohen-Addad, Esfandiari, Mirrokni, Narayanan, STOC'22]  and $6.12903$ for $k$-means [Grandoni, Ostrovsky, Rabani, Schulman, Venkat,  Inf. Process. Lett.'22]. Those are the best results, even when allowing a complexity FPT in the number of clusters $k$: this stands in sharp contrast with the $(1+\varepsilon)$-approximation achievable in that case, when allowing randomization.

In this paper, we provide deterministic sketches constructions for clustering, whose size bounds are close to the best-known randomized ones. We show how to compute a dimension reduction onto $\varepsilon^{-O(1)} \log k$ dimensions in time $k^{O\left( \varepsilon^{-O(1)} + \log \log k\right)} \text{poly}(nd)$, and how to build a coreset of size $O\left( k^2 \log^3 k \varepsilon^{-O(1)}\right)$ in  time $2^{\eps^{O(1)} k\log^3 k} + k^{O\left( \varepsilon^{-O(1)} + \log \log k\right)} \text{poly}(nd)$. In the case where $k$ is small, this answers an open question of [Feldman WIDM'20] and [Munteanu and Schwiegelshohn, Künstliche Intell. '18] on whether it is possible to efficiently compute coresets deterministically.

We also construct a deterministic algorithm for computing $(1+\varepsilon)$-approximation to $k$-median and $k$-means in high dimensional Euclidean spaces in time $2^{k^2 \log^3 k/\varepsilon^{O(1)}} \text{poly}(nd)$, close to the best randomized complexity of $2^{(k/\varepsilon)^{O(1)}} nd$ (see [Kumar, Sabharwal, Sen, JACM 10] and [Bhattacharya, Jaiswal, Kumar, TCS'18]).

Furthermore, our new insights on sketches also yield a randomized coreset construction that uses uniform sampling, that immediately improves over the recent results of [Braverman et al. FOCS '22] by a factor $k$.
\end{abstract}

\input{intro}

\section{Preliminaries}
\label{sec:prelim}
\subsection{Notations}
We use $a=(1\pm \varepsilon)b$ as an abbreviation for $(1- \varepsilon)b\leq a\leq (1+\varepsilon)b$.
For a $d$-dimensional vector $x$, we define the $\ell_z$ norm $\|x\|_z = \sqrt[z]{\sum_{i=1}^d |x_i|^z}$.
In Euclidean spaces, we let $\dist(\cdot, \cdot)$ denote the $\ell_2$ distance. 

Given a metric space $(X, \dist)$, a set of points $P\subset X$ and $\calS \subset X$, we define $\cost^z(P, \calS) = \sum_{p\in P} \min_{s \in \calS} \dist(p, s)^z$.
The $(k,z)$-clustering objective consists of finding a set of at most $k$ centers $\calS$ that minimize $\cost^z(P, \calS)$. For simplicity, we often drop the $z$ and simply write $\cost(P, \calS)$.

Another view on the problem is to find clusters rather than centers. Given a set of points $P \subset X$, we define $\mu^z(P) := \argmin_s  \cost(P, \{s\})$ to be the optimal center for $(1,z)$-clustering on $P$. For any partition $\calC = \{C_1, C_2, ...\}$ of  $P$,
we define $\costP^z(P,\calC) := \sum_{C \in \calC} \cost(C, \mu^z(C))$.

In Euclidean spaces, we say a set $P' \subset \R^{d+1}$ is an extension of a multiset $P\subset \R^d$ if for all $p\in P$, there exists $p' \in \R$ such that $P' = \left\{\zext{p}{p'}, p \in \right\}$. The number $p'$ is called the extension of $p$. If all extensions are $0$, then $P'$ is the $0$-extension of $P$.

Finally, for any function $f$ that maps points of $X$ to some other metric, we will abusively write $f(P) := \lbra f(p): p \in P\rbra$, and for a partition $\calC$, $f(\calC)$ is the partition $\lbra f(C_1), f(C_2),...\rbra$.

When clustering extensions, we will enforce centers to have extension coordinate $0$. Therefore, for a set of points $X \subset \R^{d+1}$, we define 
$\mu^z_0(X) := \argmin_{\substack{s\in \mathbb{R}^{d+1}\\ s_{d+1}=0}} \cost^z(X, \{s\})$ to be the optimal center for $(1,z)$-clustering on $P$ with the constraint that the extension coordinate is $0$, and $\costP_0(P, \calC) = \sum_{C \in \calC} \cost^z(C, \mu^z_0(C))$.

\subsection{Definitions}
We can now define the dimension reduction guarantee we aim to satisfy.

\begin{definition}[Cost Preserving Sketches for Powers]
\label{def:costpresz}
Let $z$ be a positive integer and $\varepsilon>0$ a precision parameter.
Let $P \subset \R^d$, and $f : \R^d \rightarrow \R^{m+1}$. 

Then we say that $f$ is an $(k,z,\varepsilon)$-cost preserving sketch for $P$ if for every partition $\calC$ of $P$ into $k$ sets,
\[\costP_0^z(f(P), f(\calC)) = (1\pm \eps) \costP^z(P, \calC).\]
\end{definition}

This definition captures the guarantee satisfied for $(k,z)$-clustering given by \cite{MakarychevMR19}, with the generalization of allowing $f(P)$ to have an extension.
Let us briefly compare this definition to cost preserving sketches for $k$-means, defined as follows:
\begin{definition}[Cost Preserving Sketches for $k$-means~\cite{CEMMP15}]
\label{def:costpres2}
Let $\varepsilon>0$ be a precision parameter.
Let $P \subset \R^d$, and $f : \R^d \rightarrow \R^{m}$. 

Then we say that $f$ is an $(k,2,\varepsilon)$-cost preserving sketch for $P$ if for every partition $\calC$ of $P$ into $k$ sets,
\[\costP^2(f(P), f(\calC)) + \offset = (1\pm \eps) \costP^2(P, \calC) .\]
\end{definition}

It turns out that those two definitions are equivalent for the case $z=2$: the offset $\offset$ used for $k$-means is a special case of adding extensions to $f(P)$ and requiring the extension coordinate of the centers being set to 0 (i.e. $s_{m+1}=0$). Writing $f(p) = \zext{f'(P)}{p'}$ and $\mu_0^z(C) = \zext{\mu_C}{0}$ we indeed have for $z=2$ 
$$\sum_{C \in \calC} \sum_{p \in C} \|f(p)-\mu_0^z(C)\|^2 = \sum_{C \in \calC} \sum_{p \in C} \dist(f'(p), \mu_C)^2 + {p'}^2.$$
Note that the contribution of the final coordinate is independent of the choice of centers, under the constraint $c_{m+1}=0$. Hence, the offset $\offset$ is equivalent to $\sum_{p\in P}  {p'}^2$. Second, $\mu_C$ is exactly the mean of $C$, regardless of the extension coordinates: hence, $\costP^2(f(P), f(\calC)) + \offset = \costP_0^2(f'(P), f'(\calC)) + \sum_{p \in P} P'$, and the two definitions coincides. 

In order to construct our cost preserving sketches, we will rely on partition coreset, defined as follows: 
\begin{definition}[Partition Coresets (see \cite{SohlerW18})]
\label{def:core}
Let $B$ be a multiset consisting of several copies of points in $\R^d$, with total size $|P|$.

$B$ is an $(\varepsilon,k,z)$-extension partition coreset of $P$ if there exists an extension $B'$ of $B$ and a one-to-one mapping $g : P \rightarrow B'$ such that for any partition $\calC = \{C_1, ..., C_k\}$ into $k$ parts and any centers $\calS = \{s_1, ..., s_k\}$ with $0$-extension $\{s'_1, ..., s'_k\}$, it holds that
$$\sum_{i=1}^k \sum_{p\in C_i} \dist(g(p), s'_i)^z = (1\pm \eps)\sum_{i=1}^k \sum_{p\in C_i} \dist(p, s_i)^z$$

 We say that the size of the coreset is the number of distinct points of $B$.
\end{definition}
We insist on the fact that the size of an extension partition coreset is the number of distinct points in $B$, not on the extension $B'$.

This definition of coreset is somewhat stronger than the commonly found definition that requires the cost of an assignment $\sum_{p\in A} \underset{s\in \calS}{\min} \dist(p,s)^z$ to be approximated for any set of $k$-centers $\calS$. In our definition, the points are not necessarily assigned optimally to the center set. This difference is crucial for dimension reduction purposes, as guarantees required for cost preserving sketches are related to partitions as well.

Coresets with extensions where introduced by Sohler and Woodruff~\cite{SohlerW18} for $k$-median only: our definition extends theirs, while providing a stronger guarantee, as cost is preserved for any assignments and not only assignments to the closest center.

Finally, we define the more standard coreset that we will work with in section~\ref{sec:coresetImproved}.
The definition of coreset we consider here allows for an \emph{offset}, originally due to Feldman, Schmidt and Sohler~\cite{FSS13}. For other works using this definition, see~\cite{BecchettiBC0S19,CEMMP15,SohlerW18, Cohen-AddadSS21, stoc-lb}.
\begin{definition}\label{def:stdcoreset}
  Let $(X, \dist)$ be a metric space, and $P \subset X$.
  An $(\eps, k, z)$-coreset for the $(k, z)$-Clustering problem on $P$ is a set $\coreset \subset X$ with weights $w : \coreset \rightarrow \R_+$ together with a scalar $\offset$ such that, for any set $\calS \subset X^d$, $|\calS| = k$, 
\[\cost^z(\coreset, \calS) + \offset = (1\pm \eps)\cost^z(P, \calS) \]
\end{definition}
As opposed to the previous sketches, those coresets preserve the cost of any set of \textit{centers}, instead of preserving the cost of any clustering. In particular, those centers may not be the optimal ones for the induced partition. 

Lastly, a $(\alpha,\beta)$ bicriteria algorithm for clustering produces a clustering with $\beta\cdot k$ centers and cost $\alpha\cdot OPT$, where $OPT$ denotes the cost of an optimal clustering with $k$ centers. 

\subsection{Useful Results}
\begin{lemma}[Triangle Inequality for Powers]
\label{lem:weaktri}
Let $a,b,c$ be an arbitrary set of points in a metric space with distance function $d$ and let $z$ be a positive integer. Then for any $\varepsilon>0$
$$d(a,b)^z \leq (1+\varepsilon)^{z-1} d(a,c)^z + \left(\frac{1+\varepsilon}{\varepsilon}\right)^{z-1} d(b,c)^z $$ and
$$\left\vert d(a,b)^z - d(a,c)^z\right\vert \leq \varepsilon \cdot d(a,c)^z + \left(\frac{z+\varepsilon}{\varepsilon}\right)^{z-1} d(b,c)^z.$$
\end{lemma}
Proofs of this inequality (and variants thereof) can be found in~\cite{BecchettiBC0S19,Cohen-AddadS17,FSS13,MakarychevMR19})

We will use an algorithm that compute an approximation to $(k,z)$-clustering up to a tiny error, as stated in the following theorem.  We prove the theorem
in \cref{sec:bicriteria}: although key to our algorithms, the techniques to build this approximation are distinct from our main techniques and we deferred the proof to the end.

\begin{restatable}{theorem}{bicriteria}
\label{thm:bicriteria}
There exists an algorithm running in time $k^{\log \log k + \log(1/\eps)} \cdot \poly(nd)$ that compute a set of $O(k\cdot \log(1/\eps)/\eps)$ many centers $\calS$, such that $\cost(P, \calS) \leq (1+\eps)\cost(P, \opt)$, where $\opt$ is the best solution with $k$ centers.
\end{restatable}

\paragraph{Derandomizing Random Projections.}
The seminal Johnson-Lindenstrauss lemma states that any $n$ point set in $d$-dimensional Euclidean space can be (linearly) embedded into a $m\in O(\varepsilon^{-2}\log n)$ dimensional space such that all pairwise distances are preserved up to a $(1\pm \varepsilon)$ factor. 
Most proofs show this guarantee by sampling a matrix from an appropriate distribution. We will use the following derandomization based on the conditional expectation method.

\begin{theorem}[\cite{EngebretsenIO02}]
\label{thm:EIO}
Let $v_1,\ldots v_n$ be a sequence of vectors in $\mathbb{R}^d$ and let $\varepsilon,F\in (0,1]$. Then we can compute in deterministic time $O(nd(\log n \varepsilon^{-1})^{O(1)})$ a linear mapping $S:\mathbb{R}^d\rightarrow \mathbb{R}^m$ where $m\in O(\varepsilon^{-2}\log 1/F)$ such that
\[ (1-\varepsilon)\cdot \|v_i\|_2 \leq \|S v_i\|_2 \leq (1+\varepsilon)\|v_i\|_2\]
for at least a $(1-F)$ fraction of the $i$'s.
\end{theorem}
Note that setting $F<1/n$ implies that the guarantee holds for all vectors. Also, we note that if we are given $n$ points rather than $n$ vectors and wish to preserve the pairwise distances, the running time has a quadratic dependency on $n$ as there are ${n\choose 2}$ many distance vectors. Since the running time is dominated by computing a bicriteria approximation, we will omit it from the statement of our theorems.

\section{Deterministic Partition Coresets for $(k,z)$-Clustering}
\label{sec:partitionCoreset}

\subsection{Construction of the Coreset}
The main goal of this section is to prove the following theorem:
\begin{theorem}
\label{thm:extension}
Let $A$ be a set of $n$ points in $\mathbb{R}^d$, $\eps > 0$ and $k$ be an integer. We can compute in deterministic time $k^{\log \log k + \eps^{-O(z)}} \cdot \poly(nd)$ an $(\varepsilon,k,z)$-extension partition coreset of $A$ with at most $k^{\varepsilon^{-O(z)}}$ distinct points. 
\end{theorem}

The key structural lemma we use in this case is as follows. Essentially, it states that if the cost of clustering to a single center cannot be decreased by adding $k$ centers, then the center is itself a partition coreset. 

\begin{lemma}
\label{lem:key}
Let $A$ be a set of points in $\mathbb{R}^d$ and let $m$ be a (not necessarily optimal) center for $A$. Let $M$ be the multiset of $|A|$ copies of $m$.

Suppose that for all set of $k$ centers $\calS$,
$$\cost(A,\{m\}) - \cost(A,\calS) < \alpha \cdot \cost(A,\{m\}),$$
for $\alpha\leq \frac{\varepsilon^{z+6}}{401408\cdot 2^{3z}\cdot z^{z+6}}$.

Then $M$ is an $(\eps,k,z)$-extension partition coreset of $A$.

More precisely, define the function $f$ that maps $p$ to $f(p) = \zext{m}{\|p-m\|}$, and let $M' = f(A)$. Then, for any partition $\calC = \{C_1, ..., C_k\}$ of $A$ into $k$ parts and any centers $c_1, ..., c_k$ with $0$-extension $c'_1, ..., c'_k$, it holds that
$$\sum_{i=1}^k \sum_{p\in C_i} \cost(f(p), c'_i) = (1\pm \eps)\sum_{i=1}^k \sum_{p\in C_i} \cost(p, c_i)$$
\end{lemma}

In other terms, the guarantee we have on $M$ is that fot any partition $\calC$ and any centers $c_1, ..., c_k$, it holds that
$$\sum_{i=1}^k \sum_{p\in C_i} \left(\|p - m\|^2 + \|m-c_i\|^2\right)^{z/2} = (1\pm \eps)\sum_{i=1}^k \sum_{p\in C_i} \|p-c_i\|^z.$$
Before proving this lemma, we will show why it implies the theorem.

\begin{proof}[Proof of Theorem~\ref{thm:extension}]
The algorithm to construct the partition coreset is \cref{alg:partCoreset}. In the following, let $\beta$, and $\gamma$ be constants with a polynomial dependency on $\varepsilon^{z}$ that we specify later. 

\begin{algorithm}
\caption{PartitionCoreset($C, m$, depth)}
\label{alg:partCoreset}
\begin{algorithmic}[1]
\State Compute a set $\calS$ for $(k,z)$-clustering on $C$, with cost $(1+\beta)\cost(C, \opt_C)$ and with $|\calS| = O(k \log(1/\beta)/\beta)$, using the algorithm of \cref{thm:bicriteria}.
\If{$\cost(C,\{m\})-\cost(C,\calS) \leq \beta\cdot \cost(C,\{m\})$} 
\State Return a multiset $M$ with $|C|$ copies of $m$  as a partition coreset for $C$, with extension $M' = \left\{\zext{m}{\|p-m\|}, p \in C\right\}$
\ElsIf{depth = $\gamma$}
\State Return a multiset $M$ with $|C|$ copies of $m$ as a partition coreset for $C$, with extension $M' = \left\{\zext{m}{0}, p \in C\right\}$
\Else 
\State Let $C_1, ..., C_{|\calS|}$  be the clusters induced by $\calS$, with centers $c_1,..., c_{|\calS|}$
\State Return $\cup_{i = 1}^{|\calS|} $PartitionCoreset$(C_i, c_i, $ depth$+1)$.
\EndIf
\end{algorithmic}
\end{algorithm}

We first argue that the algorithm is correct, namely, PartitionCoreset$(P, 0, 0)$ computes a partition coreset. We will then show that it runs in the desired (deterministic) time.

For the correctness, we proceed as follows. Let $(C_1, m_1, d_1), (C_2, m_2, d_2), ...$ be the set of calls to PartitionCoreset that end in line 3. Let $M_i$ be the set computed for the call with parameters $(C_i, m_i, d_i)$. We show that for any $i$, $M_i$ is a partition coreset for $C_i$.

For this, we can use Lemma~\ref{lem:key}. Let us drop the subscript $i$ and fix $C, m$ to be some $C_i, m_i$, and let $\beta = \frac{1}{2}\cdot \frac{\varepsilon^{z+6}}{401408\cdot 2^{3z}\cdot z^{z+6}}$ (i.e.: equal to $\alpha/2$ with $\alpha$ being specified in Lemma~\ref{lem:key}). Let $\calS$ be the solution computed line 1 of the algorithm for $C$: for any set of $k$ centers $K$, it holds that $\cost(C, \calS) \leq (1+\beta) \cost(C, K)$. Therefore, we have:
\begin{eqnarray*}
\cost(C,\{m\}) - \cost(C,K) &\leq & \cost(C,\{m\}) - \frac{1}{1+\beta}\cdot \cost(C,\calS) \\
&= &\frac{1}{1+\beta}\left(\cost(C,\{m\}) - \cost(C,\calS)\right) + \frac{\beta}{1+\beta}\cdot \cost(C,\{m\}) \\
&\leq &\frac{\beta}{1+\beta}\cdot \cost(C,\{m\}) + \frac{\beta}{1+\beta}\cdot \cost(C,\{m\}) \\
\implies  \cost(C,\{m\}) - \cost(C,K) &\leq & 2\beta\cdot \cost(C,\{m\}).
\end{eqnarray*}
Using Lemma~\ref{lem:key}, this implies that $M$ with its extension is indeed an $\eps$- extension partition coreset for the points in $C$.

We now deal with the other termination criteria, namely when we reach depth $\gamma$. For a level $\ell$, we define $(C^\ell_1, c^\ell_1, \ell), (C^\ell_2, c^\ell_2, \ell), ...$ be the set of calls to PartitionCoreset with level$= \ell$, and let $\Gamma = \cup C^{\gamma}_i$, $\calS_\Gamma = \cup c^\gamma_i$.
One can easily show that $\sum_i \cost(C^{\ell}_i, c^\ell_i) \leq \frac{1}{1+\beta} \sum_i \cost(C^{\ell-1}_i, c^{\ell-1 }_i)$, as if the cost does not decrease by $\frac{1}{1+\beta}$ then the recursion stops on line 3. Hence, the cost of clustering $\Gamma$ to $\calS_\Gamma$ is
$\cost(\Gamma, \calS_\Gamma) \leq  z^{O(z)}\cdot \left(\frac{1}{1+\beta}\right)^{\gamma} \opt$, 
where $\opt$ is the optimal cost for $(k, z)$-clustering on the full dataset. 
We define $\gamma$ such that  $2^{O(z)}\cdot \left(\frac{1}{1+\beta}\right)^{\gamma} \le \left(\frac{\varepsilon}{8z}\right)^{z}$: this holds when
$\gamma = \Omega\left(\frac{z\log \left(\frac{8z^2}{\varepsilon}\right)}{\log(1+\beta)}\right)$, and so $ \gamma \in O(\varepsilon^{-O(z)})$ as desired. 
For each point $p\in \Gamma$, let $c_p$ be its center in the solution $\calS_\Gamma$.
Now we consider a set of at most $k$ centers $K$, and denote $K_p$ the center to which $p\in \Gamma$ is assigned. We have
\begin{eqnarray*}
& &\left\vert\sum_{p\in \Gamma} \cost(p, K) - \cost(c_p, K)\right\vert \\
(Lem.~\ref{lem:weaktri}) &\leq & \sum_{p\in \Gamma} \varepsilon/2\cdot \cost(p, K) + \left(\frac{4z+\varepsilon}{\varepsilon}\right)^{z-1} \cdot \cost(p, c_p)\\
(\text{Cost of }\Gamma ) &\leq &  \varepsilon/2\cdot \cost(\Gamma, K) + \left(\frac{4z+\varepsilon}{\varepsilon}\right)^{z-1} \cdot  \left(\frac{\varepsilon}{8z}\right)^{z} \cdot \opt\\
 &\leq & \varepsilon/2 \cost(A, K) + \varepsilon/2 \cdot \opt \leq \varepsilon\cdot  \cost(A,K),
\end{eqnarray*}
Therefore, the set $M_\Gamma$ with its extension $M'_\Gamma = \left\{\zext{c_p}{0}, p \in \Gamma\right\}$ is indeed a partition coreset for $\Gamma$.  

By linearity of the function $\cost$, this implies that partitionCoreset$(P, 0, 0)$ indeed compute an $(\eps, z, k)$-extension partition coreset of $P$: $M_i$ is an $(\eps, z, k)$-extension partition coreset for $C_i$, $M_\gamma$ is one for $\Gamma$, and $P$ is the disjoint union of $\Gamma$ and the $C_i$.

The size of that coreset (i.e., the number of distinct points when removing the last coordinate) is the number of recursion call in the algorithm, which is at most $k^{\varepsilon^{-O(z)}}$ by an easy induction: when the level is less than $\gamma$, a set $C$ is divided into $O(k\cdot \log(1/\beta)/\beta)  = k\cdot (z/\varepsilon)^{z+O(1)}$ clusters, and there are $\gamma \in O\left(\eps^{-O(z)}\right)$ many levels.

We now consider the running time.
The algorithm of \cref{thm:bicriteria} runs in time $k^{\log \log k + \log(1/\eps)} \cdot \poly(nd)$.
It is used once per cluster of each level, hence
as above at most $k^{\eps^{-O(z)}}$ many times before terminating, and therefore a total complexity of $k^{\log \log k + \eps^{-O(z)}} \cdot \poly(nd)$.

To make explicit the polynomial in $\eps$ in the particular case where $z= 1$ (resp. $z=2$), note that $\beta$ is essentially $\eps^{z+6}$, and therefore $\gamma \approx 1/\log(1+\beta) \approx \eps^{-z-6}$. The number of recursion call and size of the coreset is therefore $k^{\eps^{-z-6}}$. 
\end{proof}

The remaining part of this section will now be devoted to proving Lemma~\ref{lem:key}.

\subsection{Proof of the Structural Lemma~\ref{lem:key}}
We first require a bit of notation. Fix a partition into $k$ parts $\calC = \{C_1,..., C_k\}$ and let $\{c_1,\ldots c_k\}$ be the centers.

For every point $p \in C_i$, we denote by $q_p$ the projection of $p$ onto the line through $m$ and $c_i$. Let $Q_i$ be the multiset of projections of the points in $C_i$. By the Pythagorean theorem, we have for any $p\in C_i$:
$$ \|p-c_i\| = \sqrt{\|p-q_p\|^2 + \|q_p-c_i\|^2}.$$
    
Our goal is to show that for ``most'' of the points $p\in A$ one of the following conditions hold.
\begin{description}
    \item[Cond. 1] $\|p-m\|\geq \left(\frac{4z}{\varepsilon}\right)\cdot \|m-c_i\|$ or $\|p-m\|\leq \left(\frac{\varepsilon}{4z}\right)\cdot \|m-c_i\|$ (i.e. the conditions of  Lemma~\ref{lem:tighterbound} below are met) or
    \item[Cond. 2] $\|q_p-m\| \leq \varepsilon/(7z)\cdot \|p-m\|$ (i.e. the conditions of Lemma~\ref{lem:projbound} below are met).
\end{description}
The following two lemmas show that if we can assume that one of these two cases hold, then the points in $M'$ corresponding to the point $p$ is a good proxy for $\|p-c_i\|^z$. We will deal afterwards in \cref{lem:badcost} with points satisfying neither of the conditions.

\begin{lemma}
\label{lem:tighterbound}
If either $\|p-m\|\geq \left(\frac{4z}{\varepsilon}\right)\cdot \|m-c_i\|$ or $\|p-m\|\leq \left(\frac{\varepsilon}{4z}\right)\cdot \|m-c_i\|$ then
$$\left\vert \left(\|p-m\|^2 + \|m-c_i\|^2\right)^{z/2} - \|p-c_i\|^z \right\vert \leq \varepsilon\cdot \left(\|p-m\|^z + \|p-c_i\|^z\right).$$
\end{lemma}
\begin{proof}
Condition 1 essentially says that $m$ is either really close to $p$ or to $c_i$. In the first case, $\|p-m\|^2$ is negligible, and $\|p-c_i\| \approx \|m-c_i\|$. In the second case, $\|m-c_i\|$ negligible and $\|p-c_i\| \approx \|p-m\|$.  We formalize this argument below.

We first show that Condition 1 implies $\|p-c_i\|^z = (1 \pm \eps / 2)\max(\|p-m\|^z,\|m-c_i\|^z)$. Indeed, we have:
\begin{eqnarray*}
\max(\|p-m\|,\|m-c_i\|) - \min(\|p-m\|,\|m-c_i\|) \geq (1-\varepsilon/(4z))\cdot \max(\|p-m\|,\|m-c_i\|) \\
\max(\|p-m\|,\|m-c_i\|) + \min(\|p-m\|,\|m-c_i\|) \leq (1+\varepsilon/(4z))\cdot \max(\|p-m\|,\|m-c_i\|)
\end{eqnarray*}
Using the Mercator series, it holds that 
$(1 + \varepsilon/(4z))^z = e^{z \cdot \ln(1 + \varepsilon/(4z))} \leq  e^{\varepsilon/4} \leq e^{\ln (1+\varepsilon/2)} = 1+\varepsilon/2$: therefore,
\begin{eqnarray*}
\|p-c_i\|^z  &\leq & (\|p-m\| + \|m-c_i\|)^z\\
& \leq & (1+\varepsilon/2)\cdot \max(\|p-m\|^z,\|m-c_i\|^z)
\end{eqnarray*}
Using Bernoulli's inequality, $(1 - \varepsilon/(4z))^z \geq 1-\varepsilon/4$: therefore,
\begin{eqnarray*}
\|p-c_i\|^z  &\geq & \lpar \max(\|p-m\|,\|m-c_i\|) - \min(\|p-m\|,\|m-c_i\|)  \rpar^z \\
& \geq & (1 - \varepsilon/(4z))^z\cdot \max(\|p-m\|^z,\|m-c_i\|^z) \\
& \geq & (1-\varepsilon/2)\cdot \max(\|p-m\|^z,\|m-c_i\|^z).  
\end{eqnarray*}

We bound similarly the other term of the expression:
\begin{eqnarray*}
    \left(\|p-m\|^2 + \|m-c_i\|^2\right)^{z/2} & \leq & \left( (1\pm \varepsilon/(4z))\cdot \max(\|p-m\|^2,\|m-c_i\|^2)\right)^{z/2}\\
    & \leq & (1\pm \varepsilon/2) \cdot \max(\|p-m\|^z,\|m-c_i\|^z)
\end{eqnarray*}

Combining those expressions, we obtain\begin{eqnarray*}
& & \left\vert \left(\|p-m\|^2 + \|m-c_i\|^2\right)^{z/2} - \|p-c_i\|^z \right\vert   \\
&= & \big\vert (1\pm \varepsilon/2) \cdot \max(\|p-m\|^z,\|m-c_i\|^z) - (1\pm \varepsilon/2)\cdot \max(\|p-m\|^z,\|m-c_i\|^z) \big\vert \\
&\leq & \varepsilon \cdot (\|p-m\|^z + \|p-c_i\|^z)  
\end{eqnarray*}
\end{proof}

\begin{lemma}
\label{lem:projbound}
Let $p,m,c_i\in \mathbb{R}^d$. Denote by $q_p$ the orthogonal projection of $p$ onto the line through $m$ and $c_i$. Then if for $0\leq \varepsilon\leq 1/2$ and  $z\geq 1$
$$\|m-q_p\|\leq \varepsilon/(7z)\cdot \|p-m\|,$$ 
we have:
$$\left\vert \left(\|p-m\|^2 + \|m-c_i\|^2\right)^{z/2} - \|p-c_i\|^z \right\vert \leq \varepsilon\cdot \|p-c_i\|^z.$$
\end{lemma}
\begin{proof}
This theorem states that when the projection of $p$ on the line between $m$ and $c_i$ is really close to $m$, then Pythagorean theorem almost hold between $p, m$ and $c_i$.  
To show this, we first require a lower bound on $\|p-c_i\|$ in terms of $\|p-m\|$. We have, using the assumption on $\|m-q_p\|$:
\begin{eqnarray}
\nonumber
\|p-c_i\|^2 &\geq & \|p-q_p\|^2 \geq \|p-m\|^2 - \|q_p-m\|^2 \\
\label{eq:projbound1}
&\geq & \|p-m\|^2 \cdot (1-\varepsilon^2/(49z^2))
\end{eqnarray}
Using  the Pythagorean theorem, we can relate $\|p-m\|^2$ to $\|p-c_i\|$ as follows: $\|p-m\|^2 = \|p-q_p\|^2 + \|q_p-m\|^2 + \|q_p-c_i\|^2 - \|q_p-c_i\|^2 = \|p-c_i\|^2 + \|q_p-m\|^2 - \|q_p-c_i\|^2$

Combined with Lemma~\ref{lem:weaktri}, we get the following: 
\begin{eqnarray*}
\|p-m\|^2 + \|m-c_i\|^2& = & \|p-c_i\|^2 + \|q_p-m\|^2 - \|q_p-c_i\|^2 + \|m-c_i\|^2 \\
&=& \|p-c_i\|^2 + \|q_p-m\|^2 \pm \left(\frac{\varepsilon}{4z}\cdot \|q_p-c_i\|^2 + \left(1+\frac{4z}{\varepsilon}\right) \|q_p-m\|^2\right)  \\
&=&\|p-c_i\|^2  \pm \left(\frac{\varepsilon}{4z}\cdot \|q_p-c_i\|^2 + \left(2+\frac{4z}{\varepsilon}\right) \|q_p-m\|^2\right)
\end{eqnarray*}
Now, using that $\|q_p - c_i\| \leq \|p-c_i\|$ (since a projection decreases distances) and the assumption on $\|m-q_p\|$, we get:  
\begin{eqnarray*}
\|p-m\|^2 + \|m-c_i\|^2&=& (1 \pm \varepsilon/(4z))\|p-c_i\|^2  \pm \left(\frac{2\varepsilon + 4z}{\varepsilon}\right) \cdot \frac{\varepsilon^2}{49z^2}\cdot \|p-m\|^2 \\
(Eq.~\ref{eq:projbound1}) & = & (1 \pm \varepsilon/(4z))\|p-c_i\|^2 \pm \left(\frac{2\varepsilon + 4z}{\varepsilon}\right) \cdot \frac{\varepsilon^2}{49z^2}\cdot \frac{1}{1-\frac{\varepsilon^2}{49z^2}} \cdot \|p-c_i\|^2 \\
(\varepsilon<\frac{1}{2}) &=& (1 \pm \varepsilon/(2z))\|p-c_i\|^2
\end{eqnarray*}
We can conclude, using again the Mercator series and Bernoulli's inequality:
\begin{eqnarray*}
\left\vert \left(\|p-m\|^2 + \|m-c_i\|^2\right)^{z/2} - \|p-c_i\|^z \right\vert&=& \left\vert (1 \pm \varepsilon/(2z))^z \cdot \|p-c_i\|^z  - \|p-c_i\|^z \right\vert \\
 &=& \left\vert (1 \pm \varepsilon)\|p-c_i\|^z  - \|p-c_i\|^z \right\vert \leq \varepsilon \cdot \|p-c_i\|^z
 \end{eqnarray*}
\end{proof}

We now bound the cost of the remaining points.
Let $D$ be the set of points that satisfy neither conditions, i.e. for every point $p\in D$ we have
\begin{equation}
\label{eq:key1}
\left(\frac{\varepsilon}{4z}\right)\cdot \|m-c_i\|\leq \|p-m\|\leq \left(\frac{4z}{\varepsilon}\right)\cdot \|m-c_i\|
\end{equation}
and
\begin{equation}
\label{eq:key2}
  \|q_p-m\| \geq \varepsilon/(7z)\cdot \|p-m\|.  
\end{equation}

Those two conditions allow us to relate the cost of clustering the points to $m$ and the cost of solution $\calS$. Using the additional property of $\calS$, we can show that the contribution of the points in $D$ to the clustering cost of $m$ is small.

\begin{lemma}
\label{lem:badcost}
Let $m$ be a point such that there exists no solution $\calS$ with $k$ centers and $\cost(A,\{m\}) - \cost(A,\calS) \leq \alpha \cdot \cost(A,\{m\})$.

Suppose $\calC = \{C_1,..., C_k\}$ is a partition with corresponding centers $\{c_1,...,c_k\}$ and let $D$ be  the set of points that satisfy neither condition (1) nor condition (2). It must be that
$$\sum_{p\in D}\|p-m\|^z < \alpha\cdot \frac{25088z^6}{\varepsilon^6}\cdot \sum_{p\in A} \|p-m\|^z.$$
\end{lemma}
\begin{proof}
The high level idea of the proof is to show the existence of a set of $k$ centers $T$ that decreases the cost of $D$ by a significant factor. This implies that the cost of $D$ must be small compared to the overall cost.

Consider a center $c_i$ and let $\ell_i = m+(m-c_i)\cdot \frac{\varepsilon^2}{28z^2}$ and $r_i=m+(m-c_i)\cdot \frac{\varepsilon^2}{28z^2}$. Let $B$ be union of all $\ell_i$ and $r_i$.
We assign each point $p\in D\cap C_i$ to its closest center $b_i\in \{\ell_i,r_i\}$. 
Note that due to Equations~\ref{eq:key1} and~\ref{eq:key2} $\|q_p-m\| \geq \frac{\varepsilon}{7z}\cdot \|p-m\| \geq \frac{\varepsilon^2}{28z^z}\cdot \|m-c_i\| = \|b_i-m\|$.
Using the Pythagorean theorem, we then have
\begin{eqnarray}
\nonumber
\|p-b_i\|^z &=& \left(\|p-q_p\|^2 + \|q_p-b_i\|^2\right)^{z/2} \\
\nonumber
&\leq & \left(\|p-q_p\|^2 + \|q_p-m\|^2 - \|b_i-m\|^2\right)^{z/2} \\
\nonumber
& = & \left(\|p-m\|^2 - \left(\frac{\varepsilon^2}{28z^2}\right)^2\cdot \|m-c_i\|^2\right)^{z/2} \\
\nonumber
(Eq.~\ref{eq:key1}) & \leq & \left(\|p-m\|^2 - \frac{\varepsilon^6}{12544z^6}\cdot \|p-m\|^2\right)^{z/2} \\
 \nonumber
& = & \|p-m\|^z\cdot (1- \frac{\varepsilon^6}{12544z^6})
\end{eqnarray}

Summing this over all points therefore leads to a cost decrease of at least
\begin{eqnarray*}
\sum_{p\in D} \|p-m\|^z - \sum_{p\in D} \|p-b_i\|^z \geq \frac{\varepsilon^6}{12544z^6}\cdot \sum_{p\in D} \|p-m\|^z.
\end{eqnarray*}

Since we use $2k$ centers, the average cost decrease per center is $\frac{1}{2k}\cdot\frac{\varepsilon^6}{800z^6}\cdot \sum_{p\in D} \|p-m\|^z$. We greedily pick the $k$ centers with maximum cost decrease. Denote this set $T$ and denote the set of points with cost decrease by $D_T$. By Markov's inequality, we therefore have 
\begin{eqnarray*}
\sum_{p\in D_T} \|p-m\|^z - \sum_{p\in D_T} \|p-b_i\|^z \geq \frac{\varepsilon^6}{25088z^6}\cdot \sum_{p\in D} \|p-m\|^z.
\end{eqnarray*}

Since we assumed that no clustering exists decreasing the cost by more than $\alpha\cdot \sum_{p\in A}\|p-m\|^z$, this implies

\begin{eqnarray*}
\frac{\varepsilon^6}{25088z^6}\cdot \sum_{p\in D} \|p-m\|^z \leq \alpha\cdot \sum_{p\in A}\|p-m\|^z \\
\Rightarrow \sum_{p\in D} \|p-m\|^z \leq \alpha \cdot \frac{25088z^6}{\varepsilon^6}\cdot \sum_{p\in A}\|p-m\|^z.
\end{eqnarray*}
\end{proof}

\begin{remark}
The dependency on $\varepsilon$ can be improved to at least $\varepsilon^{-4}$, using a more involved analysis to determine the cost of a cheaper clustering $T$. Since the subsequent analysis incurs a dependency $\varepsilon^{-z}$, we chose the simpler proof in favour of optimizing lower order terms in the exponent of $\varepsilon$.
\end{remark}

We can now conclude the proof of Lemma~\ref{lem:key}. 

\begin{proof}[Proof of Lemma~\ref{lem:key}]
Recall that for $p \in A$, we defined $f(p) = \zext{m}{\|p-m\|}$, and for any center $c_i$, $c'_i = \zext{c_i}{0}$.
We have, using the Pythagorean theorem and triangle inequality:
\begin{eqnarray}
\nonumber
& &\left\vert \sum_{i=1}^k \sum_{p\in C_i} \cost(p, c_i) - \sum_{i=1}^k \sum_{p\in C_i} \cost(f(p), c'_i)\right\vert\\
\nonumber
& = &\left\vert \sum_{i=1}^k \sum_{p\in C_i}\|p-c_i\|^z - \left(\|m-p\|^2 + \|m-c_i\|^2\right)^{z/2} \right\vert \\
\label{eq:keygood}
& \leq &\left\vert \sum_{i=1}^k \sum_{p\in C_i\setminus D}\|p-c_i\|^z - \left(\|m-p\|^2 + \|m-c_i\|^2\right)^{z/2} \right\vert \\
\label{eq:keybad}
& &+ \left\vert \sum_{i=1}^k \sum_{p\in C_i\cap D} \|p-c_i\|^z - \left(\|m-p\|^2 + \|m-c_i\|^2\right)^{z/2} \right\vert.
\end{eqnarray}
The term in Equation~\ref{eq:keygood} can be bounded using Lemmas~\ref{lem:tighterbound} and~\ref{lem:projbound} by 
$$\sum_{i=1}^k \sum_{p\in C_i\setminus D} \varepsilon\cdot (\|p-c_i\|^z + \|p-m\|^z) \leq \varepsilon\cdot \left(\sum_{i=1}^k  \cost(C_i, \{c_i\}) + \sum_{i=1}^k \cost(C_i, \{m\})\right).$$
For the term in Equation~\ref{eq:keybad}, we use a looser bound and \cref{lem:badcost} as follows: 
\begin{eqnarray}
\nonumber
    & & \left\vert \sum_{i=1}^k \sum_{p\in C_i\cap D} \|p-c_i\|^z - \left(\|m-p\|^2 + \|m-c_i\|^2\right)^{z/2} \right\vert
\\
\nonumber
&\leq & 2^{z+1} \sum_{i=1}^k \sum_{p\in C_i\cap D} \max(\|p-c_i\|^z, \|m-c_i\|^z) \\
\nonumber
(Eq.~\ref{eq:key1})&\leq & 2^{z+1} \sum_{p\in D} \left(1+\left(\frac{4z}{\varepsilon}\right)^z\right)\cdot \|p-m\|^z  \\
\label{eq:keyfinal}
(Lem.~\ref{lem:badcost}) &\leq &  2^{z+2}\cdot \left(\frac{4z}{\varepsilon}\right)^z \cdot \alpha\cdot \frac{25088z^6}{\varepsilon^6} \cdot \sum_{p\in A}\|p-m\|^z
\end{eqnarray}
Combining with \cref{lem:tighterbound} and \cref{lem:projbound}, we then obtain: 
\begin{eqnarray*}
\nonumber
& &\left\vert \sum_{i=1}^k  \cost(C_i, \{c_i\}) - \sum_{i=1}^k  \cost(f(C_i), \{c'_i\})\right\vert \\
\nonumber
(Lem.~\ref{lem:tighterbound} \text{ and } \ref{lem:projbound})&\leq & \varepsilon\cdot \left(\sum_{i=1}^k  \cost(C_i, \{c_i\}) + \sum_{i=1}^k \cost(C_i, \{m\})\right) \\
(Eq.~\ref{eq:keyfinal}) & &+  2^{z+2}\cdot \left(\frac{4z}{\varepsilon}\right)^z \cdot \alpha\cdot \frac{25088z^6}{\varepsilon^6}\cdot \cost(A, \{m\})\\
(\text{no cost decrease})&\leq & 2\cdot\left(2\varepsilon  +  2^{z+2}\cdot \left(\frac{4z}{\varepsilon}\right)^z \cdot \alpha\cdot \frac{25088z^6}{\varepsilon^6}\right) \cdot \sum_{i=1}^k  \cost(C_i, \{c_i\})
\end{eqnarray*}

The overall result now follows by our choice of $\alpha$ and rescaling $\varepsilon$.
\end{proof}

\section{Derandomized Dimension Reduction}
\label{sec:dimred}

In this section we will use the deterministic construction of partition coresets to obtain a derandomized dimension reduction for $(k,z)$ clustering. The goal is to show that we can obtain a cost preserving sketch with dimension at most $O(\varepsilon^{-O(z)}\log k)$:

\begin{theorem}\label{thm:dimred}
Let $P \subset \R^d$. 
We can compute an $(k,z,\varepsilon)$-cost preserving sketch of $P$ with target dimension $O\left(\varepsilon^{-O(z)}\cdot \log k\right)$ in deterministic time $k^{\log \log k + \varepsilon^{-O(z)}} \cdot \poly(nd)$.
\end{theorem}
In the particular $k$-median case ($z=1$), the precise target dimension is $\eps^{-11} \log k$, while it is $\eps^{-12} \log k$ for $k$-means. We did not make any particular attempt to optimize the dependency in $\eps$.

Again, the first-thought approach would be to use the deterministic Johnson-Lindenstrauss transform of \cite{EngebretsenIO02} (Theorem~\ref{thm:EIO}). Unfortunately, merely preserving pairwise distances is not sufficient to preserve the cost to an optimal $(1,z)$-center, except for the special case $z=2$.
Instead, we introduce the following notion.

\begin{definition}[Witness Sets]
\label{def:witness}
Let $P$ be a set of $n$ points in $\mathbb{R}^d$ (with a possible extension coordinate), let $c$ be an optimal $(1,z)$-center (under the possible constraint that the extension coordinate is 0) and let $\Delta=\frac{\sum_{p\in P}\|p-c\|^z}{n}$.
Then a $(D,R,\varepsilon)$-uniform witness set is a subset $S\subset P$ such that
\begin{itemize}
    \item The convex hull of $S$ contains a $(1+\varepsilon)$-approximation of $(1,z)$-clustering on $P$ (under the possible constraint that the extension coordinate is 0).
    \item The diameter of $S$ is at most $D\cdot \sqrt[z]{\Delta}$.
    \item The size of $S$ is at most $R$.
\end{itemize}
\end{definition}

The first important result is proving that the existence of witness sets of small size exist for a specific range of parameters. 

\begin{theorem}
\label{thm:witness}
There exists $\left(\frac{O(z)}{\varepsilon},\varepsilon^{-2}\cdot 2^{O(z)},\varepsilon\right)$ witness sets.
\end{theorem}

We give a full proof of Theorem~\ref{thm:witness} in Section~\ref{sec:witness} of the appendix. The existence of small witness sets is essentially a consequence of the existence of small coresets. Let us consider how this existence of small witness sets combined with Theorem~\ref{thm:extension} implies that we can deterministically construct a cost preserving sketch for powers.

The main idea is that the property of witness sets allows us
to efficiently ``discretize'' all candidate solutions in the convex hull of a witness set. 
We will show that the dimension reduction preserves the property of a subset of points being a witness set. Furthermore, because witness sets lie in low-dimensional spaces, we can guarantee that the cost for every point in the convex hull of a witness is preserved. Specifically, we use that linear projections preserve subspaces, and in particular the subspace spanned by the points in the convex hull of the witness set. Hence, if the cost of every candidate solution in a witness set stays the same, up to a $(1\pm \varepsilon)$ factor, the clustering cost overall also stays the same.

\begin{lemma}\label{lem:partitionDimRed}
Let $P \subset \R^d$ be a a multiset with at most  $T$ distinct points.
There exists a set $\calN$ of size $T^{O \lpar\varepsilon^{-2}\rpar}$ that can be computed in deterministic time $T^{O \lpar\varepsilon^{-2}\rpar}$ with the following property. If a linear projection $\Pi$ preserves pairwise distances between points of $\calN$ up to a $(1\pm \varepsilon/z)$ error, then:

for any extension $P'$ of $P$ where each point $p$ has coordinate extension $p'$,  and for any subset $C$ of $P$ (with extension $C'$), 
\[\sum_{p \in C}\left\|\zext{\Pi p}{p'}- \mu^z_0(\Pi'(C))\right\|^z = (1\pm \eps)\sum_{p \in C} \left\|\zext{p}{p'} - \mu^z_0(C')\right\|^z,\]
where $\Pi'(C) := \left\{\zext{\Pi p}{p'}, p\in C\right\}$ is the projection of $C$ with the same extensions as $C'$.
\end{lemma}
\begin{proof}
First, $\calN$ is constructed as follows. Let $\varepsilon' = \varepsilon/(4Dz)$, where $D = O(z)/\eps$ is the first parameter of the witness set provided in \cref{thm:witness}.
For each possible set $S \subseteq P$ of size $O(1/\eps^2)$, construct an $\eps' \cdot \diam(S)$-cover of the convex hull of $S$.\footnote{Recall that a $\varepsilon$-cover of a set $S$ is a set of points $N$ such that for an $x\in S$, we have some $y\in N$ with $\|x-y\|\leq \varepsilon$.} 
Moreover, let $V$ be an orthogonal basis of the subspace spanned by $S$, with the addition of the origin. 

$\calN$ is made of the union of those covers and basis, for all possible sets $S$. To bound its size, observe that each set $S$ lies in an $|S|=O\left(\varepsilon^{-2}\right)$-dimensional flat. Hence, we can compute an $\varepsilon' \diam(S)$-cover of $S$ of size at most $\left(\frac{O(\sqrt{|S|})}{\varepsilon'}\right)^{|S|} = \eps^{-O(\eps^{-2})}$ in deterministic time polynomial in $\eps^{-O(\eps^{-2})}$~\cite{Cha01}. 
Similarly, we can compute the basis $V$ in deterministic polynomial time.

Since there are at most $T^{O(\varepsilon^{-2})}$ many different set $S$, the size of $\calN$ is therefore bounded by 
$$T^{O\left(\varepsilon^{-2}\right)} \cdot \eps^{-O(\eps^{-2})} = T^{O\left(\varepsilon^{-2}\right)}.$$

We now verify that $\calN$ has the desired property. Essentially, preserving the norm of each basis $V$ will ensure that the operator norm of $\Pi$ is bounded, and preserving the distances to all net points, combined with the existance of small witness set, will guarantee that the clustering cost is preserved.

Let $P'$ be an extension of $P$ as in the lemma statement, and define $\calN'$ to be the set of $0$-extensions of $\calN$. 
We first note that any possible witness set $S'$ for $P'$ under the constraint that the extension coordinate is $0$, an $\eps' \diam(S')$-cover of $\conv(S') \cap \{x \in \R^{d+1}: x_{d+1} = 0\}$ is contained in $\calN'$. 
Indeed, let $S$ be the corresponding witness set for $P$ (i.e., $S'$ where we removed the extension coordinate): it holds that $\conv(S') \cap \{x \in \R^{d+1}: x_{d+1} = 0\} \subseteq \conv(S)$, and $\diam(S) \leq \diam(S')$, so the $\eps' \diam(S)$ cover of $\conv(S)$ is also an $\eps' \diam(S')$-cover of $\conv(S') \cap \{x \in \R^{d+1}: x_{d+1} = 0\}$.

Let $\Pi$ be a linear projection such that the distance from any point of $P$ to any point of $\calN$ is preserved up to an error $(1\pm \eps/z)$,  and $C$ be a subset of $P$. As in the lemma statement, we define $\Pi'(C) := \left\{\zext{\Pi p}{p'}, p\in C\right\}$.
What we need to show is that $\costP_0(\Pi'(C)) = (1\pm \eps)\costP_0(C')$ or, in other terms, 
\[   \sum_{p \in C} \left\|\zext{p}{p'} - \mu^z_0(C')\right\|^z= (1\pm \eps) \sum_{p\in C}\left\|\zext{\Pi p}{p'} - \mu^z_0(\Pi'(C))\right\|^z,\]

We first prove the somewhat easier direction which states that the cost does not increase with the projection. Let $S'$ be a $(D, R, \eps')$-witness set for $C'$ given by \cref{thm:witness}.
Denote by $c^*$ the point whose $0$-extension is in the convex hull of $S'$ which is the $(1+\varepsilon)$-approximate solution to $\mu^z_0(C')$, and 
$c$ the point of $\calN$ closest to $c^*$. We will show that $\zext{\Pi c}{0}$ is a good center for $\Pi'(C)$.

Since $\calN$ contains an $\eps'\diam(S)$-cover of $\conv(S)$ and  $S$ is a witness set, we have
\begin{equation}
    \label{eq:derandom1}
    \|c^*-c\| \leq \varepsilon'\cdot \diam(S) \leq \varepsilon' \cdot D \cdot \left(\frac{\costP_0(C')}{|C|}\right)^{1/z} 
\end{equation}

Let $\Delta_{C} := \frac{\costP_0(C')}{|C|}$. 
We have, for any point $p$:
\begin{eqnarray*}
\|\Pi p-\Pi c\|& & \\
(\text{Distortion of Embedding})&\leq & (1+\varepsilon/z)\cdot \|p-c\| \\
&\leq & (1+\varepsilon/z)\cdot \left(\|p-c^*\| + \|c^*-c\|\right) \\
(Eq.~\ref{eq:derandom1}) &\leq & (1+\varepsilon/z) \cdot \left(\|p-c^*\| +\varepsilon'\cdot D \cdot \Delta_{C}^{1/z} \right) \\
(\text{choice of }\eps') & \leq & (1+\varepsilon/z) \cdot \left(\|p-c^*\| +\eps/(2z) \cdot \Delta_{C}^{1/z} \right) \\
\end{eqnarray*}
Therefore, using \cref{lem:weaktri}, we get $\|\Pi p-\Pi c\| \leq (1+\eps/z)^3 (\|p-c^*\|^2 + \frac{\eps}{2z} \Delta_C^{2/z})$. Since the embedding of $c$ is $\zext{\Pi c}{0}$, we obtain: 
\begin{align*}
\left\|\zext{\Pi p}{p'} - \zext{\Pi c}{0}\right\|^z &= \left(\|\Pi p - \Pi c\|^2 + p'^2\right)^{z/2}\\
&\leq \left((1+\varepsilon/z)^3 \cdot \left(\|p-c^*\|^2 +\frac{\eps}{2z}\cdot (\Delta_{C})^{2/z}\right)  + {p'}^2\right)^{z/2}\\
&\leq (1+\varepsilon/z)^{3z/2} \cdot\left( (\|p-c^*\|^2 + {p'}^2) +\frac{\eps}{2z}\cdot (\Delta_{C})^{2/z}\right)^{z/2}\\
&\leq (1+\varepsilon/z)^{3z/2}\left( (1+\eps) \cdot (\|p-c^*\|^2 + {p'}^2)^{z/2} + \left(\frac{z+\eps}{\eps}\right)^{z/2-1} \left(\frac{\eps}{2z}\right)^{z/2} \cdot \Delta_{C}\right)\\
\qquad &\leq (1+3\eps) \cdot (\|p-c^*\|^2 + {p'}^2)^{z/2} + \eps \cdot \Delta_{C}
\end{align*}

Summing over all $p \in C$, we get that
\begin{align}
    \notag 
    \costP_0(\Pi'(C)) &\leq \sum_{p\in C} \left\|\zext{\Pi p}{p'} - \zext{\Pi c}{0}\right\|^z \leq  (1+3\eps) \costP_0(C') + \eps \costP_0(C')\\
    &\leq (1+O(\eps))\costP_0(C').\label{eq:dimred1}
\end{align}

The proof of the other direction is similar in spirit, but quite more technical. For ease of reading, we defer it to the Appendix and \cref{lem:halfPartitionDimRed}.
\end{proof}

Combining with the existence of small partition coresets, the previous lemma concludes the proof of \cref{thm:dimred}.  Essentially, it is enough to compute the set $\calN$ and apply \cref{lem:partitionDimRed} on a extension partition coreset of $P$, which has therefore size independent of $n$. Due to the presence of extensions, the proof requires a bit of technical attention.

\begin{proof}[Proof of Theorem~\ref{thm:dimred}]
The cost preserving sketch for $P$ is constructed as follows. First, use \cref{thm:extension} to compute an $(\eps, k, z)$-extension partition coreset $\coreset$ of $P$, with extension $\coreset'$, in time $k^{\log \log k + \varepsilon^{-O(z)}} \cdot \poly(nd)$. Let $g$ be the corresponding mapping from $P$ to $\coreset$ (as in \cref{def:core}).

To apply \cref{lem:partitionDimRed} on $\coreset$, we start by computing the set $\calN$ given in the statement, for the set $\coreset$, in time $|\coreset|^{O(\eps^{-2})} = k^{\eps^{-O(z)}}$. The size of $\calN$ is as well $k^{\eps^{-O(z)}}$.
We the use the algorithm from \cref{thm:EIO} to deterministically compute a linear mapping $\Pi:\mathbb{R}^d\rightarrow \mathbb{R}^m$ with $m\in \varepsilon^{-O(z)}\log k $ such that pairwise distances between points of $\calN$ are preserved, up to a distortion factor $(1\pm \varepsilon/z)$. This can be done in time $O\left(n\cdot d\cdot k^{\varepsilon^{O(-z)}}\cdot \log^{O(1)} n\right)$. 

Now, \cref{lem:partitionDimRed} ensures that for the extension $\coreset'$ of $\coreset$ (where each point $p$ has coordinate extension $p'$),  and for any for any subset $C$ of $\coreset$ (with extension $C'$), 
\begin{gather}\label{eq:propPi}
    \sum_{p \in C}\left\|\zext{\Pi p}{p'}- \mu^z_0(\Pi'(C))\right\|^z = (1\pm \eps)\sum_{p \in C} \left\|\zext{p}{p'} - \mu^z_0(C')\right\|^z,
\end{gather}
where $\Pi'(C) := \left\{\zext{\Pi p}{p'}, p\in C\right\}$ is the projection of $C$ with the same extensions as $C'$.

We define our cost preserving sketch for $P$ as follows: for any $p \in P$, let $f(p) := \zext{\Pi g(p)}{g(p)'}$, where $g(p)'$ is the extension of $g(p)$ in $\coreset'$. 
Let $\calC = \lbra C_1,...,C_k \rbra$ be a clustering of $P$.
We will show that $\cost(f(P), f(\calC)) = (1\pm \eps) \cost(P, \calC)$, which will conclude. 

For any cluster $C_i$, we denote $g(C_i)$ the multiset $\lbra g(p) : p \in C_i \rbra$, and $g(C_i)'$ the extension $\lbra \zext{g(p)}{g(p)'} : p \in C_i \rbra$. 
\begin{eqnarray*}
     \costP_0(f(P), f(\calC)) &=& \sum_{i=1}^k \cost(f(C_i), \mu^z_0(f(C_i)))\\
     &=& \sum_{i=1}^k \sum_{p \in C_i}\left\|\zext{\Pi g(p)}{g(p)'}- \mu^z_0(\Pi'(g(C_i)))\right\|^z\\
 \text{(using \cref{eq:propPi})} 
    & \leq & (1+\eps)\sum_{i=1}^k \sum_{p \in C_i}\left\|\zext{g(p)}{g(p)'} - \mu^z_0(g(C_i)')\right\|^z\\
\text{(by def. of }\mu^z_0)    &\leq& (1+\eps)\sum_{i=1}^k \sum_{p \in C_i}\left\|\zext{g(p)}{g(p)'} - \zext{\mu^z(C_i)}{0}\right\|^z\\
   \text{(by prop. of $g$)} 
   & \leq& (1+\eps)^2 \sum_{i=1}^k \sum_{p\in C_i} \|p - \mu^z(C_i)\|^z\\
    & =& (1+\eps)^2 \costP(P, \calC).
\end{eqnarray*}

We proceed similarly for the other direction: let $s_i$ such that $\zext{s_i}{0} = \mu_0^z(g(C_i)')$
\begin{eqnarray*}
     \costP(P, \calC) &=& \sum_{i=1}^k \sum_{p \in C_i}\left\|p- \mu^z(C_i)\right\|^z\\
     &\leq & \sum_{i=1}^k \sum_{p \in C_i}\left\|p- s_i\right\|^z\\ 
     \text{(by prop. of $g$)} 
     & \leq& (1+\eps) \sum_{i=1}^k \sum_{p\in C_i} \left\|\zext{g(p)}{g(p)'} - \mu_0^z(g(C_i)')\right\|^z\\
     \text{(using \cref{eq:propPi})} 
     &\leq& (1+\eps)^2 \sum_{i=1}^k \sum_{p \in C_i}\left\|\zext{\Pi g(p)}{g(p)'}- \mu^z_0(\Pi'(g(C_i)))\right\|^z\\
     &=& (1+\eps)^2\sum_{i=1}^k \cost(f(C_i), \mu^z_0(f(C_i)))\\
     &=& (1+\eps)^2 \costP_p(f(P), f(\calC)) 
\end{eqnarray*}

Hence, $\costP_0(f(P), f(\calC))  \in (1\pm O(\eps))\costP(P, \calC)$: $f$ is indeed a cost preserving sketch for $P$, with target dimension $m+1 = \eps^{-O(z)} \log k$, which concludes the proof.

In the particular case where $z \in \{1, 2\}$, then the size of the partition coreset is essentially $k^{\eps^{-z-6}}$: therefore, the size of $\calN$ is $k^{\eps^{-z-8}}$, and the target dimension of the embedding is $\eps^{-z-10} \log k$.
\end{proof}

\begin{remark}
If allowing randomization, one can compute an arbitrary JL embedding onto $\log k \eps^{-O(z)}$ dimensions: with good probability, this will preserves distances between points of $\calN$ defined in \cref{lem:partitionDimRed}, and therefore the proof of Theorem~\ref{thm:dimred} goes through with this embedding. In particular, one could use the sparse-JL constructions of \cite{KaneN14} to get a faster running time in case of sparse inputs.
\end{remark}

\input{coreset}

\input{detFPT}

\section{Acknowledgements}
D. Sauplic has received funding from the European Union’s Horizon 2020 research \erclogowrapped{5\baselineskip} and innovation
programme under the Marie Skłodowska-Curie grant agreement No 101034413, and Grant agreement No.\ 101019564
``The Design of Modern Fully Dynamic Data Structures (MoDynStruct)".
\newline C. Schwiegelshohn acknowledges the support of the Independent Research Fund Denmark (DFF) under a Sapere Aude Research Leader grant No 1051-00106B.

\newpage
\bibliography{references}{}
\bibliographystyle{alpha}

\newpage

\input{appendix}

\end{document}

%% file: intro.tex
\section{Introduction}
Clustering is the task of partitioning a dataset in a meaningful way such that similar data elements are in the same part and dissimilar ones in different parts. Clustering problems are among the most studied problems in theoretical computer science. 
Clustering applications range from data analysis~\cite{jajuga2012classification} to compression~\cite{lloyd1982least}, and this has given raise to several different ways of modeling the problem.
The $k$-median and $k$-means problems are among the most prominent clustering objectives: algorithmic solutions to these problems are now part of the basic data analysis toolbox, and used as subroutines for many machine learning procedures \cite{lin2009phrase, coates2011analysis, coates2012learning}. 
Furthermore, $k$-median and $k$-means have deep connections with other classical optimization problems, and understanding their complexity has been a fruitful research direction: Their study has led to many interesting algorithmic development (for instance for primal-dual algorithms \cite{JaV01, LiS13, BPRST15} or local-search based ones \cite{Cohen-AddadGHOS22, soda23}).

Given a metric space $(X, \dist)$, and two sets $P, \cand \subseteq X$, the goal of the $k$-median problem (respectively $k$-means problem) is to find a set of $k$ centers $\calS \in \cand^k$ in order to minimize the sum of distances (resp. distances squared) for each point in $P$ to its closest center in $\calS$.
The complexity of the $k$-median and $k$-means problems, and the algorithmic techniques used to solve them, depends naturally on the underlying metric space. 
Several works have thus focus on Euclidean inputs that are arguably the most common case in practice, and exhibits deep algorithmic challenges.
Indeed, in that case the set of candidate centers $\cand$ is infinite, and even computing the solution to $1$-median -- also called the Fermat-Weber problem -- is not easy: there is no closed form formula, and many works investigated the complexity required to approximate the best center \cite{ChandrasekaranT89, ChinMMP13, Badoiu02, CLMPS16}.

Most of the algorithmic tools to handle the Euclidean version of $k$-median and $k$-means are heavily reliant on randomization, and are Monte-Carlo algorithms\footnote{A Monte-Carlo algorithm has a deterministic running time and a fixed probability of success.}:
We do not know whether the algorithm succeeds or not, hence, we cannot be sure of the result. This remark also holds for sketching 
techniques for which it is unclear whether the sketch has the 
desired properties.

We propose to study deterministic algorithms for clustering problems.
This is interesting from both practical and theoretical standpoints: from a practical perspective, it makes sure that the results can be reproduced independently of the underlying system or programming language\footnote{Obtaining similar probabilistic outcomes over different architectures, computers, programming languages is a challenge.}, and that they are correct; 
from a theoretical standpoint, understanding the power of randomness in algorithms is a very basic question and a recurring theme in theoretical computer science. 
We ask in particular the following question:

\begin{quest}\label{quest:apx}
Is it possible to deterministically compute 
$(1+\eps)$-approximation for $k$-median and $k$-means, in time competitive with the best randomized result?
\end{quest}

For high-dimensional Euclidean spaces, the best known algorithm has complexity $2^{(k/\eps)^{O(1)}} dn$ \cite{KumarSS10}. In deterministic time however, only constant-factor approximation are known, as for instance the primal-dual algorithm of
\cite{AhmadianNSW17}. 

To tackle those clustering problems in high-dimensional spaces, one of the most important tools is sketching,
which has been a rich and vibrant area of research for the past ten years. 
The efforts to design even smaller sketches have resulted in an impressive success story producing algorithms with performance guarantees, having high impact in practice \cite{MaaloufJF19, MunteanuSSW18, LiebenweinMFR21, feldman2016dimensionality}. 
There are two natural ways of reducing the size of a $k$-median or $k$-means input: First by computing a smaller set of points $B$ that approximates the cost of any $k$-means (or $k$-median) solution of $A$. Such sets are known as \emph{coresets}. 
The second way consists in finding a projection of the input point sets onto some low-dimensional space such that the $k$-means (or $k$-median) cost of any partitioning is approximately preserved.
 In this context, the quality of a sketch is usually quantified in terms of its size, i.e.: the dependency on the size of the input $n$, the dimension $d$ of the underlying Euclidean space (when there is one), the number of centers $k$, and a precision parameter $\varepsilon$. 
For both objectives, we now have optimal, or nearly optimal results, see \cite{Cohen-AddadSS21, stoc-lb} for coresets constructions and lower bounds and \cite{MakarychevMR19, LarsenN17}  for dimension reduction.
These works are the results of a tremendous effort that has finally 
led to dimensionality reduction and 
coreset bounds of size independent on $n$ (the number of input points) and the dimension $d$ for the Euclidean case.

The applications for sketches go well beyond designing approximation algorithms: they are at the heart of modern data analysis, allowing to process efficiently big data and to design algorithms for models of computation where the memory is constrained. 
For instance, they allow to design efficient algorithms for the streaming model, distributed computing or massively-parallel computations (MPC).

Interestingly -- and somewhat frustratingly -- these state-of-the-art results are heavily reliant on randomization. 
Concretely, to achieve a success probability of $1-\delta$, coreset constructions require $\Omega_{k,\eps}(\log (1/\delta))$ points, and dimension reduction techniques require $\Omega_{k,\eps}(\log (1/\delta))$-dimensional spaces. 
Therefore, to construct a sketch (either a coreset or a dimension reduction) that offers the desired guarantees with \textit{high probability} (i.e., probability $1-\poly(1/n)$ where $n$ is the original input size), the dependency in $n$ strikes back. At the same time, randomized constructions imply that there exist sketches whose size is independent on $n$. Therefore, we ask:

\begin{quest}\label{quest:sketch}
Is it possible to deterministically construct $k$-median and $k$-means sketches of size matching the best randomized results?
\end{quest}

Finding deterministic coresets of small size is explicitly asked as an open question in recent surveys on the coreset literature~\cite{feldmanSurvey, MunteanuS18}.

The best deterministic constructions for Euclidean $k$-median yields coresets of size 
$O(poly(k)\cdot \varepsilon^{-(d+O(1))})$~\cite{HaK07}, which compares 
very unfavourably to the best randomized result of size $\tilde O(k 
\eps^{-3})$ from \cite{stoc-lb}. In particular, it suffers from the 
classic \emph{curse of dimensionality} and it is not clear a priori 
whether such a dependency is needed when removing randomization.
A notable exception is for $k$-means, where it is possible to leverage 
the algebraic structure of the problem to construct coresets of size $ k^{\eps^{-2}\log(1/\eps)}$~\cite{FeldmanSS20}. Unfortunately, this construction is very specific to $k$-means and does not apply to $k$-median, and requires additional effort to be implemented deterministically. Importantly, the construction yields a coreset 
whose size is very far from the best known randomized bounds.

Suppose we were to relax the problem of deterministically computing a coreset and aim for a Las-Vegas algorithm.
\footnote{Recall that a Las Vegas algorithm returns a solution in expected polynomial time, whereas a Monte Carlo algorithm only returns a solution with a small probability of failure.} For closely related problems such as locality sensitive hashing this is achievable (see Ahle~\cite{Ahle17} and references therein). For clustering, all known algorithms are Monte Carlo algorithms. The main barrier to obtaining a Las Vegas algorithm (a challenging and interesting open problem as well) 
is that we currently do not know how to verify that a candidate set of points is indeed a valid coreset without evaluating all possible choices of $k$ centers and verifying that the cost of each set of $k$ centers is preserved. This is also a major roadblock for using conditional-expectation methods to derandomize.
For some tiny $\eps$, the problem of verifying that a set is a coreset is even co-NP-hard \cite{chris-esa}. This stands in stark contrast with other related techniques for which derandomization techniques have been designed such as the Johnson Lindenstrauss lemma~\cite{EngebretsenIO02}, matrix approximation~\cite{BoutsidisDM11}, or regression~\cite{Sarlo2006}.

\subsection{Our Results}

We phrase our results in terms of the $(k,z)$-clustering problems, a generalization of both $k$-median and $k$-means: the cost function is $\cost(P, \calS) = \sum_{p\in P} \min_{s \in \calS} \dist^z(p, s)$, and the goal is to find the solution $\calS$ that minimizes the cost. Thus, $z=1$ is the $k$-median problem, $z=2$ is $k$-means.

Our main technical contribution lies in designing sketches to answer \cref{quest:sketch}. We start by showing a dimension reduction result:
\begin{theorem}[Informal, see \cref{thm:dimred}]
Let $P$ be a set of $n$ points in $\R^d$. We can compute in time $k^{\eps^{-O(z)} +  \log \log k} \poly(nd)$ a sketch of $P$ onto $\eps^{-O(z)} \cdot \log k$ dimensions that preserves the cost of any $(k,z)$-clustering.
\end{theorem}

The optimal target dimension, even using randomization, is $\Omega \lpar \eps^{-2} \log k\rpar$: our result has therefore an optimal dependency in $k$. The best known deterministic result was $O \lpar \eps^{-2} \log n\rpar$ (implicit in the work of \cite{MakarychevMR19}), with a dependency in $n$ that we precisely want to avoid.

Building on this dimension reduction result, we show how to construct coresets deterministically: 
\begin{theorem}[Informal, see \cref{cor:main-coreset}]
Let $P$ be a set of $n$ points in $\R^d$. We can compute in time $2^{\eps^{-O(z)} k\cdot \log^3 k} + k^{\eps^{-O(z)} + \log \log k}\poly(nd)$ a set $\coreset \subset \R^{d+1}$ of size $O\lpar k^2 \log^2 k \eps^{-O(z)}\rpar$ such that, for any set of $k$ centers $\calS \subset \R^d$, 
\[\cost(\coreset, \calS') = (1\pm \eps)\cost(P, \calS),\]
where $\calS' := \lbra \zext{s}{0}, s \in \calS \rbra$ is the set $\calS$ where we append a 0 coordinate at each center.
\end{theorem}

This notion of coreset is used in \cite{SohlerW18}, and is more general than the most standard one (for which $\coreset \subset P$). However, it maintains the most interesting property that any $\alpha$-approximation for $\coreset$ with all last coordinate equal to zero is an $(1+\eps)\alpha$-approximation for $P$. Therefore, any good solution on $\coreset$ yields a good solution for $P$, and we can compute that way a $(1+\eps)$-approximation.

Those sketches directly imply deterministic streaming algorithm for $(k,z)$-clustering, using the merge-and-reduce technique. The required memory is merely $O \lpar k^2 d \log ^{O(z)} n \eps^{-O(z)}\rpar$, where the best randomized result requires memory $O\lpar kd \log^2 n \eps^{-O(z)}\rpar$ to compute a $(1+\eps)$-approximation \cite{BravermanFLR19}. This also apply to distributed computing and MPC models, matching the best randomized result up to $\poly(\log n, \eps^{-1})$ factors.

In case $k$ is large (namely $\log k \log \log k \geq \log n$), we could replace the complexity by $n^{O\lpar \eps^{-2}\rpar}$. However, one application of the sketches is to answer positively \cref{quest:apx}, for which having a complexity in terms of $k$ is more desirable.
We show:
\begin{theorem}\label{thm:apx}
For any fixed $z\geq 1, \eps>0$, there exists a deterministic algorithm that, given a set $P$ of $n$ points in $\R^d$, computes a $(1+\eps)$-approximation to $(k,z)$-clustering on $P$ in time $k^{\log \log k + \varepsilon^{-O(z)}} \cdot \poly(nd)  + 2^{k \log^3 k \cdot \eps^{-O(z)}}$. 
\end{theorem}

Guruswamy and Indyk  \cite{GuI03} showed it is NP-hard to compute an approximation arbitrarily close to $1$ in dimension $O(\log n)$ (and for arbitrary $k$): hence, to have a polynomial running time in the dimension, an exponential dependency in $k$ is necessary. Our theorem essentially matches the best randomized running time of \cite{KumarSS10} for $k$-median and $k$-means, of $2^{(k/\eps)^{O(1)}} nd$.

\subsection{Implication of our work when allowing randomization.}
We believe that the sketches we develop have interest beyond computing a $(1+\eps)$-approximation. Indeed, we develop new sketching techniques, that improve over the best known bounds for some regime of parameters, and that are of independent interest.

We remark that, if allowing randomization, our new sketch constructions can be implemented faster and provide useful alternatives to the existing ones. For dimension reduction, the randomized running time of our procedure is the one to compute a Johnson-Lindenstrauss (JL) transform, in general $\tilde O(nd)$. However, we can use \emph{any} JL transform to get our result, including the sparse transforms (e.g., \cite{KaneN14}) that allow for further speed-ups. This stands in contrast with the \cite{MakarychevMR19} result, which needs a particular additional property on the JL transform they use, namely sub-gaussian tails, and therefore cannot use the fastest JL embeddings.

More importantly, our coreset construction can be implemented using only uniform sampling. The importance of uniform sampling was illustrated in the recent paper \cite{focs22}, as it gives an easier way of constructing coreset. 
Our construction improves the one of \cite{focs22} by a factor $k$: if one can compute a coreset of size $T$ for any \emph{ring} $P \subset \lbra x: \|x\| \in [1/2, 1]\rbra$, then for any input one can compute a coreset of size $O\lpar T \cdot k \log(1/\eps)/\eps\rpar$ -- as opposed to $O\lpar T k^2 / \eps \rpar$ in \cite{focs22}.

This apply to any metric where the VC-dimension of the uniform function space can be bounded, including shortest-path metric of minor-free graphs and Fréchet distance. 

\subsection{Our Techniques}
\label{sec:techniques}

Given a coreset for $(k,z)$-clustering, it is quite straightforward to compute a $(1+\eps)$-approximation: one can simply enumerate all partitions of the coreset into $k$ parts, evaluate their cost and keep the best partition. The running time is thus $\min(k^{\text{coreset-size}}, \text{corest-size}^{kd})$, where the latter bound can be obtained through the result of \cite{inaba1994applications}.

Therefore, we focus on the construction of a coreset. For that, we first require a dimension reduction tool, that we explain first.

\subsubsection{Dimension Reduction}

\paragraph*{High-level overview.}
There currently exist only two known methods that enable deterministic dimension reduction. The first is principal component analysis and its generalizations to other norms. Applying principal component analysis to $k$-means requires a target dimension of $\Omega(k)$ in order to preserve the $k$-means cost of each partition \cite{CEMMP15} up to a constant factor and thus this attempt falls short of the $O(\log k)$ dimension we are aiming for.

The only known alternative is the conditional expectation method for derandomizing Johnson Lindenstrauss transforms by Engebretsen, Indyk and O'Donnell \cite{EngebretsenIO02}.
At a high level, this method selects random bits of the projection one by one, and selects the value of the next bit so as to minimize the number of ``bad events", namely the number of vectors whose lengths are distorted by a too high factor, conditioned on the bit's value. 

When applying this idea naively to a candidate embedding for $k$-clustering, the number of ``bad events" becomes the number of clusterings for which the cost is not  preserved.
As there exist $n^{dk}$ many distinct clusterings (\cite{inaba1994applications}), this is infeasible in time polynomial in $n,d,k$.
Thus, our goal is to reduce the number of bad events necessary to count, in order to use the conditional expectation method. 

For this, our first contribution is to reduce the number of distinct clusterings we need to test for. More precisely, we
show that for any cluster, a $(1+\varepsilon)$ approximate center lies in the convex hull of a subset of the cluster of size $O\lpar\eps^{-2}\rpar$, that we call a \textit{witness set}. 
Using a careful discretization argument, we can show that it is possible to enumerate over all candidate solutions in the convex hulls of such witness sets and thereby also to enumerate over all possible centers induced by any clustering. This reduces the clustering problem with arbitrary centers in $\R^d$ to a discrete one, where centers are restricted to a finite set, and it is therefore enough to preserve the distances from input points to the discrete set of centers.
There are $n^{O\lpar\varepsilon^{-2}\rpar}$ many possible witness sets -- any set of $O\lpar\eps^{-2}\rpar$ input points. Hence, counting the ``bad events" now  requires to test $n^{O\lpar\varepsilon^{-2}\rpar}$ many clusterings. This is still infeasible in our target running time, but, crucially, does not depend in the dimension $d$ anymore.

To improve further, the idea is to reduce the number of distinct points to $\poly(k, \eps^{-1})$, in order to reduce the number of possible witness sets. A \textit{coreset} is the most natural idea here. However the typical guarantees of  coresets are not strong enough for our purposes: as explained in \cref{sec:prelim}, coresets preserve the cost to any set of $k$ centers, while our goal here is to preserve the cost of any partition. In turn, we naturally introduce the notion of \textit{partition coreset}, which are small sets preserving the cost of any partition. 
We show how to construct deterministically partition coresets of size $k^{\varepsilon^{-O(z)}}$; hence, using the witness sets of this partition coreset, we are able to limit the number of clusterings to $k^{O\lpar\varepsilon^{-2}\rpar}$. We can therefore list all the potential bad events and count their number, and reduce the dimension to $\varepsilon^{-O(z)} \log k$ -- log of the number of distances to be preserved.
We point that the notion of partition coreset, more general then the usual coreset definition, may be of independent interest and have further application in the future. 

\paragraph*{Construction of partition coreset.} To build a partition coreset, a first attempt would be to use a $k'$ clustering such that its cost is at most an $\varepsilon$-fraction of the cost of an optimal $k$-clustering. In that case, the cost of snapping each input point to its closest center would be negligible, and the cost of any partition would be preserved. Unfortunately, such a set has size at least $\varepsilon^{-d}$
in the worst case.
Instead, we show -- and this is our main technical contribution -- that there exists a $k'$ clustering such that for each cluster, either the cluster is very cheap, or its cost cannot decrease by adding additional centers. We show such a statement for $k' = k^{\eps^{-O(z)}}$. The guarantee we ask on the $k'$ clustering is somewhat non-standard, as the cost in each cluster should not decrease by adding $k$ centers (as opposed to the global cost). 
In the first case, the set of centers directly gives a partition coreset. The second case implies essentially that the points in some cluster of the $k'$ clustering pay the same cost in \textit{any} possible $k$-clustering. 
Hence, we can simply replace those points by their center, and encode the cost difference in an extension coordinate.

To illustrate  further our argument, suppose that we are given a center $m$ such that the cost of clustering to $m$ is nearly the same as clustering to the optimum $k$ centers. Then, we show that in \emph{any} clustering into $k$ parts with centers $c_1,...,c_k$, for almost all points $p$ of $c_i$'s cluster, one of the following condition is verified: 
\begin{itemize}
    \item either $p$ is really close to $m$ compared to $c_i$, namely $\|p-m\| \leq \eps \|m-c_i\|$, 
    \item or $c_i$ is really close to $m$ compared to $p$, $\|m-c_i\| \leq \eps \|p-m\|$
    \item or the segment between $p$ and $m$ is nearly orthogonal to the one between $m$ and $c_i$.
\end{itemize}

Each of those conditions implies that $\|p-c_i\| \approx \lpar \|m-c_i\|^2  + \|p-m\|^2 \rpar ^{1/2}$: we can therefore replace the point $p$ by $f(p) \in \R^{d+1}$, where the first $d$ coordinates are the ones of $m$ and the \emph{extension coordinate} is $\|p-m\|$. Now, for any possible center $c_i$, the distance from $p$ to $c_i$ is roughly the one from $f(p)$ to $\zext{c_i}{0}$ ($c_i$ with an extension coordinate $0$). 
Therefore, the cost of clustering to $\{c_1, ..., c_k\}$ before the transformation is the same as clustering to $\lbra \zext{c_1}{0}, ..., \zext{c_k}{0}\rbra$ after the transformation. 
Furthermore, after transformation, all points have the same first $d$ coordinates (equal to $m$), and all centers are required to have last coordinate zero. 
Therefore, we have a small partition coreset for those points satisfying one of the condition.

We deal separately with points satisfying none of them. Using the fact that the cost of clustering to $m$ is nearly the same as to clustering to $k$ centers, we can show that those points have a very tiny cost to $m$, and are easier to deal with.

Applying that construction to all clusters of the $k'$ clustering whose cost do not decrease by adding $k$ centers, we get a set of points in $\R^{d+1}$ where the first $d$ coordinates are one of the $k'$ centers. 
From this, we manage to build a set of candidate centers capturing all clustering, of size $k'^{O(\eps^{-2})}$ -- henceforth independent of $n$. This concludes our algorithm: To reduce the dimension and preserve the cost of all clustering, it is enough to preserve distances between points of the partition coreset and the discrete set of centers, i.e., between $k^{\eps^{-O(z)}}$ many points, which can be done using a derandomized Johnson-Lindenstrauss
projection.

The last missing piece to the argument is to compute effectively this  $k'$ clustering deterministically.
This is not straightforward either, as it requires to compute a $(1+\eps)$-approximation to $(k,z)$-clustering. The standard techniques to do so rely on dimension reduction: using ours leads to a loop in the argument. 
Instead, we use another JL construction (specifically, the one  of Kane and Nelson~\cite{KaneN14}) that requires few random bits. 
Enumerating over all possible random bits provides a a (somewhat) small set of projections, such that at least one is a good dimension reduction for clustering. 
For each of the projections, we can compute a $(1+\eps)$-approximate solution in the projected space, and then lift it back to the original $\R^d$. For the projections that do not preserve the cost, those solutions in $\R^d$ may be bad: however, for the good projection, this solution will be a $(1+\eps)$-approximation in the original metric as well.  Enumerating over all possible solutions computed that way allows therefore to compute a $(1+\eps)$-approximation. 
It is then possible to use this algorithm to compute the desired $k'$ clustering, which concludes our whole algorithm.

\subsubsection{Coresets}
We show a generic coreset construction, based on a careful preprocessing of the input and VC-dimension techniques. This can be applied to the Euclidean space, using our partition coresets and dimension reduction, to obtain even smaller coresets. 

Building on the work of Chen~\cite{Chen09}, we show that constructing a coreset reduces to sampling uniformly. For this, we essentially compute a bi-criteria approximation and partition points in exponential rings according to their cost in the bi-criteria. There are $O(k \log n)$ rings, and Chen showed that uniform sampling among them yield a coreset. Our key new ingredient here is to show how to go from $O(k\log n)$ to only $O(k)$ rings.
For this, our key contribution is to start from a structured bi-criteria approximation, in the same spirit as what we did for the partition coreset. In particular, we ensure that the cost of that solution cannot decrease by adding $k$ more centers. We infer from that property that the cost of the points ``far" from the bi-criteria centers cannot be much different in any solution. 
We show that we can replace them by copies of the center, and encode their cost in an additive offset $\offset$.
This structural observation introduces a clean and elegant way of removing the dependency in $n$ from the coreset size.  
Similarly to what is done in \cite{Cohen-AddadSS21} (although their techniques are deeply randomized), this informally hints that the ``far" points do not affect clustering algorithms.

Next, we observe that instead of sampling points in each ring, it is enough to build an $\eps$-approximation of a particular set system, formed by $k$ balls. 
It is well known that the VC-dimension of that set system is bounded by $O(kd\log k)$ in Euclidean space of dimension $d$: hence, using a standard algorithm, this leads to $\eps$-approximation of size $\tilde O(\frac{kd \log k}{\eps^2})$.

Unfortunately, the running time of these constructions is exponential in the VC dimension, which is $k\cdot d\log k$ in our case. To eliminate the dependency on $d$, we apply first our construction of partition coreset: this reduces the size of the input to $k^{\varepsilon^{-O(z)}}$. We can then embed the input into a low dimensional space via terminal embeddings. 
Terminal embeddings are commonly used as a coreset preserving dimension reduction, see~\cite{BecchettiBC0S19,BravermanJKW21,Cohen-AddadSS21,HuV20} and the construction by Mahabadi, Makarychev, Makarychev and Razenshteyn~\cite{MahabadiMMR18} can be made deterministic. That way, we are able to obtain coresets of size $O(k^2\cdot \varepsilon^{-O(1)})$ in time $n^{O(\varepsilon^{-2}\log 1/\varepsilon)}+exp(\text{poly}(k,\varepsilon^{-1}))$.

\subsection{Related Work}

\paragraph{Approximation of $(k,z)$-clustering}
In general discrete metric spaces, $k$-median is NP-hard to approximate better than $1+2/e$, and $k$-means better than $1+8/e$. Those two bounds can be obtained by FPT algorithms running in time $f(k,\eps)\poly(n)$ \cite{Cohen-AddadG0LL19}. In polynomial time, the best approximation for $k$-median has been very recently improved to $2.613$~\cite{soda23,gowda22}, and $6.1290$ for $k$-means \cite{GrandoniORSV22}.

In Euclidean spaces, the picture is quite different as there are infinitely many possible centers location. The problems are NP-hard even for $d=2$ or $k=2$. The best approximation ratio are $2.406$ for $k$-median and $5.912$ for $k$-means, in polynomial time \cite{Cohen-AddadEMN22}. 

However, it is possible to compute $(1+\eps)$-approximations: either in time $f(\eps, d) n \polylog(n)$ \cite{Cohen-AddadFS21} or in time $2^{(k/\eps)^{O(1)}} nd$ by \cite{KumarSS10}. We show how to match the latter running time, deterministically. 

\paragraph{Coresets}
The introduction essentially covered the most important deterministic coreset constructions. The state-of-the art for $k$-median an $k$-means is a size $O(k^3\cdot \varepsilon^{-d-1})$~\cite{HaK07}, or $O(k\cdot \varepsilon^{-d-1}\log n)$~\cite{HaM04}.
Other examples for deterministic constructions can be found in~\cite{BravermanJKW19,FGSSS13,FrahlS2005,HJV19}. These and the aforementioned algorithms  were only claimed for $k$-median and $k$-means, but in our view could easily be extended to arbitrary powers (with the exception of the $k^{O(\varepsilon^{-2}\log 1/\varepsilon)}$ size coreset for $k$-means by~\cite{FeldmanSS20}). The literature on sampling based coreset algorithms for clustering is vast, see~\cite{FeldmanL11,FSS13,HuV20,SohlerW18,Cohen-AddadSS21,Cohen-AddadLSSS22,Chen09,HuangJLW18} and references therein. As for lower bounds, any coreset for $k$-means and $k$-median in Euclidean space must have size $\Omega\left(k \eps^{-2}\right)$~\cite{stoc-lb}.

Other notable problems for which we have deterministic coresets include extent approximation problems~\cite{AHV04,Chan04}, the minimum enclosing ball problem~\cite{BaC03,BHPI02}, determinant maximization~\cite{IndykMGR20,MahabadiIGR19}, page rank\cite{LangBSTFR19}, diversity~\cite{CeccarelloPP20,IndykMMM14}, subspace approximation~\cite{FSS13}, and spectral sparsification~\cite{BatsonSS09}. 

\paragraph{Dimension Reduction}
Most dimension reduction efforts have been on the $k$-means problem. A long line of research has studied the relationship between $k$-means and principal component analysis / singular value decomposition, see~\cite{BoutsidisM13,DrineasFKVV04,FSS13} and most notably Cohen, Elder, Musco, Musco and Persu~\cite{CEMMP15}, who showed that projecting onto the subspace spanned by the first $\lceil k/\varepsilon\rceil$ principal components approximates the cost of any clustering up to a $(1+\varepsilon)$-factor and that this bound is tight. 

The alternative to PCA-based algorithms are random projections and terminal embeddings. For the former, \cite{EngebretsenIO02} proposed a derandomization scheme with target dimension $O(\varepsilon^{-2} \log n)$, where $n$ is the number of vectors to be preserved. For terminal embeddings, the algorithm proposed by Mahabadi, Makarychev, Makarychev and Razenshteyn~\cite{MahabadiMMR18} achieves a target dimension of $O(\varepsilon^{-4}\log n)$ and is deterministic. The guarantee of a terminal embedding is somewhat stronger than that of \cite{EngebretsenIO02}: given a fixed set of points $X \in \mathbb{R}^d$, a terminal embedding is a mapping of $\mathbb{R}^d$ to $\mathbb{R}^m$ such that for any $y\in \mathbb{R}^d$ and any $x\in X$, the pairwise distance between $y$ and $x$ is preserved. At the same time, terminal embeddings are not linear, which can be a drawback. Other highly related deterministic algorithms exist for neareset neighbor search in high dimension, see~\cite{AbdullahAKK14,Bentley75,ColeGL04,Har-PeledIM12,Indyk03} and references therein, as well as sparse recovery~\cite{IndykP11,LiN20,NelsonNW12}.

The literature for randomized algorithms is vast, see~\cite{BecchettiBC0S19,CEMMP15,ElkinFN17,FKW19,KerberR15,NarayananN19,SohlerW18}, and most notably the recent and nearly optimal result by Makarychev, Makarychev and Razenshteyn~\cite{MakarychevMR19} for works with special bearing on $k$-clustering.

\subsection{Roadmap}
After introducing useful definitions and result in \cref{sec:prelim}, we present our construction of partition coreset in \cref{sec:partitionCoreset}. We first apply it to construct a deterministic dimension reduction in \cref{sec:dimred}, and then combine the results to construct small coresets in \cref{sec:coresetImproved}.
Finally, we present in \cref{sec:bicriteria} the bicriteria approximation algorithm, used as a building block in the previous parts.

%% file: coreset.tex
\section{Improved Deterministic Coreset Construction via Uniform VC-dimension}\label{sec:coresetImproved}

We focus in this section on a different notion of compression, namely coresets. As mentionned in the introduction, our construction works in general metric spaces, and not simply in Euclidean Spaces. In the following, we will work in a metric space  $(X, \dist)$: we will show a general coreset construction, as well as a tailored version of it to Euclidean Spaces.

Recall the definition of a coreset \cref{def:stdcoreset}: an $(\eps,k,z)$-coreset  with offset for the $(k,z)$-Clustering problem for $P$ is a set $\coreset$ with weights $w : \coreset \rightarrow \R_+$ together with a constant $\offset$ such that, for any set $\calS \subset X$, $|\calS| = k$, 
\[\cost^z(\coreset, \calS) + \offset = (1\pm \eps)\cost^z(P, \calS)\]

We will heavily rely on notions from VC-dimension literature, defined as follows:

\begin{definition}\label{def:vcdim}
  A tuple $(X, \calR)$ is a \emph{range space} when $\calR = (R_1, ..., R_m)$ with $R_i \subseteq X$.
  
  A set $A$ is an $\eps$-set-approximation\footnote{this is usually simply called an $\eps'$-approximation, but we already use this terminology for the cost of solutions.} for $(X, \calR)$ if 
  $$\forall R  \in \calR, \left| \frac{|R|}{|X|} - \frac{|R \cap A|}{|A|}\right| \leq \eps.$$ 
  
  The VC-dimension of a range space is the size of the largest set $Y$ such that $\{Y \cap R,~R\in \calR\} = 2^{|Y|}$, meaning that every subset of $Y$ is of the form $Y\cap R$ for some $R \in \calR$.
\end{definition}

The main result of the section is the following theorem. We relate the size of the coreset to the VC dimension of the range space $\calB$ formed by $k$ balls, defined as follows:
For a point $c\in X$ and $r\in \mathbb{R}_{\geq 0}$, let $H_{c,r} = \{p\in P~|~ \dist(p,c)\geq r\}$. 
For $k$ centers $C=\{c_1,c_2,\ldots ,c_k\}$, we similarly define $H_{C,r} = \{p\in P~|~ \min_{c\in C} \dist(p,c)\geq r\}$.
We let $\calB:=\bigcup_C \bigcup_r H_{C,r}$.

\begin{theorem}\label{thm:smallcoreset}
Let $(X, \dist)$ be a metric space such that the VC dimension of $(P, \calB)$ is $D$.
There exists a deterministic algorithm running in deterministic time
$\tilde O(|X|\cdot |P| \cdot k/\eps) + k \left(D/\eps\right)^{O(D)}\cdot |P|$ that constructs an $\eps$-coreset with offset for the set $P$, with size $2^{O(z\log z)}\cdot \frac{kD  \log(D) \polylog(1/\eps)}{\eps^5}$.

In the Euclidean Space $\R^d$, the running time is
$k^{\log \log k + \varepsilon^{-O(z)}} \cdot \poly(|P|d)+  k \left(D/\eps\right)^{O(D)}\cdot |P|$.
\end{theorem}

The exponential term in $D$ comes from the computation of $\eps$-set approximations: if allowing for randomization, this can be done merely by uniform sampling. In that case, the complexity becomes simply $\tilde O \lpar |X| \cdot |P| \cdot k/\eps \rpar$.

\vspace{0.5em}

\subsection{Application to Euclidean Spaces via Dimension Reduction, and Proof of \cref{thm:apx}.}
\vspace{0.5em}

Our prominent example is the Euclidean case. It is well known from the coreset literature that the set system aforementioned has VC-dimension $O(kd \log k)$ in Euclidean spaces of dimension $d$. This has been proven for example Lemma 1 of~\cite{BachemLH017} and Corollary 34 of~\cite{FeldmanSS20}~\footnote{The bounds in the literature are often stated for a generalization known as the weighted function space induced by $k$-clustering in Euclidean spaces, which is necessary for coreset constructions based on importance sampling. Here, we show how to use uniform sampling instead, which  allows us to use \cref{thm:chazelle} as a black box.}.  This in particular leads to $\eps'$-set-approximation of size $O\left(\frac{kd \log k}{\eps'^2}\cdot \log(kd \log k / \eps')\right)$.

In order to remove the dependency in $d$ and match the best coreset result, we aim at using a classical terminal embedding to reduce the dimension.
 In particular, we will start from the extension-coreset of size $k^{\eps^{-O(z)}}$ from \cref{thm:extension}, and apply a deterministic terminal embedding construction on it. 
 Since the pre-image of a coreset computed after a terminal embedding is a coreset in the original space (see \cite{HuV20} or \cite{Cohen-AddadSS21}), this yields the following corollary:

\begin{cor}\label{cor:main-coreset}
Given a set of $n$ point $P \subset \R^d$, there exist a deterministic algorithm running in time 
$k^{\log \log k + \varepsilon^{-O(z)}} \cdot \poly(nd) + \exp\lpar \eps^{-O(z)} k \log^3 k\rpar$
that constructs a set $\coreset \subseteq \R^{d+1}$ of size $k^2 \log^2 k \cdot \eps^{-O(z)}$ together with a constant $\offset$ such that:
\[\forall \calS \in (\R^d)^k, \cost(\coreset, \calS') + \offset = (1\pm \eps) \cost(P, \calS),\]
where $\calS'$ is the zero extension of $\calS$, namely, $\calS' = \lbra \zext{s}{0}, s \in \calS\rbra$.
\end{cor}
\begin{proof}
To reduce the dimension, we use our extension-partition coreset construction combined with a terminal embedding, as for instance in Huang and Vishnoi~\cite{HuV20}. 

We start by computing a $(\eps, k, z)$-extension partition coreset $\coreset_1$ of $P$ of size $k^{\eps^{-O(z)}}$, using \cref{thm:extension}, in time $k^{\log \log k + \varepsilon^{-O(z)}} \cdot \poly(nd)$.

In particular, we have for any solution $\calS \in (\R^d)^k$ with zero-extension $\calS'$
\begin{equation}
\label{eq:coreset1}
\cost(\coreset_1, \calS') = (1\pm \eps) \cost(P, \calS)
\end{equation}

Now, we compute a terminal-embedding $\Pi$ of $\coreset_1$ onto $\eps^{-O(z)} \log k$ dimensions. 
Using the techniques from \cite{HuV20}, any coreset of $\Pi(\coreset_1)$ after projection is a coreset of $\coreset_1$.

The terminal embedding can be done deterministically in time $k^{\eps^{-O(z)}}$ using the algorithm from \cite{MahabadiMMR18} combined with \cref{thm:EIO}. 

Now, we can compute a coreset of the embedding using \cref{thm:smallcoreset}. Since in dimension $O(\eps^{-O(z)} \log k)$ the VC-dimension of $(\Pi(\coreset_1),\calB)$ is at most $\eps^{-O(z)}k \log^2 k$, this is done deterministically in time 
$k^{\log \log k + \varepsilon^{-O(z)}} \cdot \poly(nd) +  \exp\lpar \eps^{-O(z)} k \log^{3} k\rpar$.
The pre-image $\coreset$ of that coreset is, by the techniques of Huang and Vishnoi, a coreset for $\coreset_1$ in $\R^{d+1}$. Therefore, we have:
\begin{equation}
\label{eq:coreset2}
\cost(\coreset, \calS') + \offset = (1\pm \eps) \cost(\coreset_1, \calS')
\end{equation}

Combining \cref{eq:coreset1} and \cref{eq:coreset2} concludes the proof.
\end{proof}

A direct application of that corollary allows to compute a $(1+\eps)$-approximation to $(k,z)$-clustering in deterministic FPT time, and proves our main \cref{thm:apx}.

\begin{proof}[Proof of \cref{thm:apx}]
The algorithm is simply the following: compute the dimension reduction of \cref{thm:dimred} onto dimension $d'$, then the coreset $\coreset$ provided by \cref{cor:main-coreset}, enumerate all possible $k$ clustering of it, keep the best one and compute the center in $\R^d$ associated with this clustering.

That way, we compute the optimal solution for $\coreset$ in $\R^{d'+1}$ that has all its last coordinates equal to $0$, $\calS'$. Let $\calS$ be the corresponding solution in $\R^{d'}$. Furthermore, let $\opt$ be the optimal solution for $P$, with zero-extension $\opt'$. We have:
\begin{align*}
(1 - \eps) \cost(P, \calS) &\leq \cost(\coreset, \calS') + \offset\\
&\leq \cost(\coreset, \opt')+\offset\\
&\leq (1+ \eps)\cost(P, \opt)
\end{align*}
Hence, $\cost(P, \calS) \leq \frac{1+\eps}{1-\eps} \cost(P, \opt)$, and $\calS$ is indeed a $(1+O(\eps))$-approximation.

To enumerate all possible $k$ clustering, we use the algorithm of \cite{inaba1994applications} who showed that for a set of $m$ points in $\R^{d'}$ there are only $n^{O(kd')}$ many possible clustering -- and provided an efficient way of doing so.

The running time is therefore: $k^{\log \log k + \varepsilon^{-O(z)}} \cdot \poly(nd) + \exp\lpar \eps^{-O(z)} k \log^3 k\rpar$ to compute $\coreset$, 
$|\coreset|^{k \log(k) \eps^{-O(z)}}$ to enumerate all $k$ clustering of $\coreset$ using \cite{inaba1994applications}, and then $\poly(k d' /\eps) = \poly(k/\eps)$ to compute the cost of each solution. Finally, the running time to compute the optimal clusters in $\R^d$ given the optimal partition in $\R^{d'}$ is $\poly(nd)$. Hence, in total, the running time is
\[k^{\log \log k + \varepsilon^{-O(z)}} \cdot \poly(nd)  + 2^{k \log^3 k \cdot \eps^{-O(z)}},\]
which concludes.
\end{proof}

\vspace{0.5em}

\subsection{Application to Minor-Excluded Graphs, Graphs with Bounded  Treewidth, and Clustering of Curves.}
\vspace{0.5em}

Besides Euclidean spaces, 
\cref{thm:smallcoreset} can be applied to metrics for which the VC-dimension of balls is bounded. In that way, it improves over the recent result \cite{focs22}, by saving a factor $k$ in the construction of coreset via uniform sampling, and allowing deterministic implementations.

One such class of metric spaces is the ones induced by graphs excluding a minor $H$: in that case, Bousquet and Thomass{\'e}~\cite{bousquet2015vc} showed that the VC-dimension of $(P, \calB)$ is $O(|H| k\log k)$.\footnote{Formally, they show the result for $k=1$. \cite{BravermanFLR19} showed that extending the result to any $k$ looses a $k \log k$ factor.} Hence applying directly \cref{thm:smallcoreset} yields a coreset of size $\tilde O\lpar\frac{k^2 |H| \log^2 k}{\eps^{5}}\rpar$. 
The construction is simplified compared to \cite{Cohen-AddadSS21}, at the price of losing a factor $k \eps^{-1}$ (or $\eps^{-3}$). On the other hand, the dependency in the size of the minor is greatly improved: it is not even specified precisely in \cite{Cohen-AddadSS21}, although it seems to be at least doubly exponential. 
This also improve over the $\tilde O\lpar |H|\poly(k/\eps)\rpar$ of \cite{focs22}, by polynomial factors.
{\setlength{\emergencystretch}{2.5em}\par}

Since graphs with treewidth $t$ exclude the complete graph with $t$ vertices, the result allows to construct coreset of size  $\tilde O\lpar\frac{k^2 t \log^2 k}{\eps^{5}}\rpar$ for those -- which needs to be compared with the randomized bound $\tilde O\lpar kt \cdot \min\lpar \eps^{-2} + \eps^{-z}, \eps^{-2}k\rpar \rpar$ of \cite{Cohen-AddadSS21}.
As we will see, the running time is dominated by the construction of an $\eps$-set-approximation. Giving up on the determinism, this can be made fast using random sampling (see Chazelle~\cite{Cha01}): this yields a running time $\tilde O(|X| \cdot |P| \cdot k/\eps)$.
Again, the construction is greatly simplified compared to \cite{Cohen-AddadSS21}, to the price of a a slightly worse coreset size.

\paragraph{Clustering of Curves.}
Driemel, Nusser, Phillips, and Psarros showed in~\cite{driemel2021vc} VC dimension bounds for clustering of curves, under the Fr{\'e}chet and Hausdorff distances. 
The input of that problem is a set of curves, each consisting of $m$ segments in $\R^d$, and the centers are restricted to be curves consisting of $\ell$ segments in $\R^d$. 

They consider two types of distances for curves: the Hausdorff distance between two curves $X$ and $Y$ is defined as follows: let $d_{\vec{H}}(X,Y) = sup_{x\in X} \inf_{y\in Y} \|x-y\|$. The Hausforff distance between $X$ and $Y$ is $\max(d_{\vec{H}}(X,Y), d_{\vec{H}}(Y,X))$. The Fr{\'e}chet distance is slightly more intricate: for that, we consider parametrized curves, namely a curve is a function $X :  [0,1] \rightarrow \R^d$ whose graph is made of $\ell$ consecutive segments. The Fr{\'e}chet distance between two parametrized curves $X$ and $Y$ is defined as $\min_{f : [0,1] \rightarrow [0,1]} \max_{t \in [0, 1]} \|X(f(t)) - Y(t)\|$, where the function $f$ is restricted to be continuous, non decreasing and with $f(0) = 0, f(1) = 1$.

For both those distances, Driemel, Nusser, Phillips and Psarros~\cite{driemel2021vc} showed that the VC dimension of $(P, \calB)$ is bounded by $O\left(k d^2 \ell^2 \log(d \ell m) \log k\right)$.
Hence, \cref{thm:smallcoreset} shows the existence of a coreset of size $\frac{k^2 d^2 \ell^2 \log (d \ell m) \log^2 k}{\eps^{O(1)}}$. This improves over \cite{buchin2021coresets} and \cite{focs22}, although our construction is not polynomial time in that case.

\vspace{0.5em}
\subsection{A Sketch of the Construction}
\vspace{0.5em}

As explained in the introduction (\cref{sec:techniques}), our construction is inspired by the one of Chen~\cite{Chen09} and of \cite{Cohen-AddadSS21}. As both of those result, we start by computing a constant-factor approximation, and placing rings of exponential radius around each of its center.

We first use a careful preprocessing step to reduce the number of such rings from $O(k \log n)$ as in \cite{Chen09} to $O(k)$. For this, we use a particularly structured constant-factor approximation, that allows us to simply discard the outer rings.

To build coresets in the remaining rings, both \cite{Chen09} and \cite{Cohen-AddadSS21} use uniform sampling. 
Instead, we show that building an $\eps'$-set-approximation of the set system formed by $k$ balls is actually enough: this is doable deterministically, which concludes our algorithm.

The size of an $\eps'$-set-approximation can be bounded via the VC-dimension of the set system: 
\begin{theorem}[Theorem 4.5 in Chazelle~\cite{Cha01}]
\mbox{}\label{thm:chazelle}
Let $(X, \calR)$ be a range space of VC-dimension $d$. Given any $r > 0$, an $r$-set-approximation for $(X, \calR)$ of size $O\left(d r^{-2} \log (d/r)\right)$ can be computed deterministically in time $d^{O(d)} \left(r^{-2} \log (d/r)\right)^d |X|$.
\end{theorem}

\begin{theorem}[Theorem 4.9 in Chazelle~\cite{Cha01}]
Let $(X, \calR)$ be a range space of VC-dimension $d$. Given any $r > 0$, a uniform sample of size $O \lpar r^{-2} (d\log (dr) + \log(1/\delta))\rpar$ is a $r$-set-approximation for $(X, \calR)$ with probability $1-\delta$.
\end{theorem}

We will compute $\eps$-set-approximation for the set system $(P, \calB)$, where $P$ is the set of input points and $\calB$ is defined in the opening paragraph.

\subsection{The algorithm}
As sketched, the algorithm partitions the input points into rings, and computes then $\eps$-set-approximation in each ring. We state here some useful definitions.

Fix a metric space $(X, \dist)$, positive integers $k, z$ and a set of clients $P$. For a solution $\calS$ of $(k,z)$-clustering on $P$ and a center $c \in \calS$, $c$'s cluster consists of all points closer to $c$ than to any other center of $\calS$. Note that this is different from the dimension reduction setting of \cref{sec:dimred}, where the assignment are not necessarily towards the closest center.

Fix as well some $\eps > 0$, and let $\constantApprox$ be any solution for $(k, z)$-clustering on $P$ with $k$ centers. 
Let $C_1, ..., C_{k}$ be the clusters induced by the centers of $\constantApprox$.

\begin{itemize}
\item the average cost of a cluster $C_i$ is $\Delta_{C_i} = \frac{\cost(C_i, \constantApprox)}{|C_i|}$
\item For all $i, j$, the \textit{ring} $R_{i,j}$ is the set of points $p \in C_i$ such that 
\[2^j \Delta_{C_i} \leq \cost(p, \constantApprox) \leq 2^{j+1} \Delta_{C_i}.\] 
\item The \textit{inner ring} $\inner(C_i)$ (resp. \textit{outer ring} $\out(C_i)$) of a cluster $C_i$ consists of the points of $C_i$ with cost at most $\left(\nicefrac \eps z\right)^z \Delta_{C_i}$ (resp. at least $\left(\nicefrac z \eps \right)^{2z} \Delta_{C_i}$), i.e., $\inner(C_i) := \cup_{j \leq z\log(\eps/z)} R_{i,j}$ and $\out(C_i) := \cup_{j > 2z\log(z/\eps)}R_{i,j}$. The \textit{main ring} $\main(C_i)$ consists of all the other points of $C_i$. For a solution $\calS$, we let $\inner^\calS$ and $\out^\calS$ be the union of inner and outer rings of the clusters induced by $\calS$.
\end{itemize}

The final algorithm is as follows:

\textbf{Input:} A metric space $(X, \dist)$, a set $P \subseteq X$, $k, z > 0$, $\eps$ such that $0 < \eps < 1/3$, and $\eps' \leq \eps$.\\
\textbf{Output:} A coreset with offset. Namely, a set of points $P' \subseteq X$, a weight function $w: P' \mapsto \R_+$
and an offset value $\offset$ such that for any set of $k$ centers $C$, $\cost(P, C) = (1\pm\eps) \cost(P', C) + \offset$. 
\begin{enumerate} 
\item \textbf{Greedy seeding:} 
\begin{enumerate}
\item Compute a $c_\constantApprox$-approx $\constantApprox$ to the $(k, z)$-clustering problem for $P$.\footnote{Any solution with $O(k)$ clusters and
cost $O(\opt)$ can be actually used for this step, in order to have a faster algorithm.}
\item Initialize $\seeded := \constantApprox$.
\item While: $\cost(\seeded) \geq \eps \cost(\constantApprox) / c_\constantApprox$ 
  and there is a candidate center $c$ such
  that
  \\
  $\left(1-\frac{\eps}{k \cdot c_\constantApprox}\right)\cost(\seeded) \ge \cost(\seeded \cup \{c\})$ do:\\
  $\seeded \gets \seeded \cup \{c\}$.
\end{enumerate}

\item If $\cost(\seeded) \leq \eps\cost(\constantApprox)/c_\constantApprox$,
  then output the following coreset and stop:
  $\seeded$ and $w : \seeded \mapsto \R_+$ where $w(c)$ corresponds to the number of
  points served by center $c$ in $\seeded$, and offset $\offset = 0$. 
  
\item\label{step:offset} Let $\offset := \cost(\out(\seeded), \seeded)$ be the cost of the outer points of $\seeded$.  
\item Set the weights of all the centers of $\seeded$ to 0.
\item\label{step:ig} \textbf{Reducing the number of rings:}:\\ For each cluster $C$ with center $c$ of $\seeded$, remove all the points in $\inner(C) \cup \out(C)$  and increase the weight of 
  $c$ by the number of removed points.
  Let $I_G$ be the instance created with those weights. 
\item \textbf{Sampling from the remaining rings:} For every cluster $C_i$, and every $j$ , compute an $\eps'$-set-approximation $\coreset_{i,j}$ of the set system $(R_{i,j}, \calB)$. 
Weight the points of $\coreset_{i,j}$ by $\frac{R_{i,j}}{|\coreset_{i,j}|}$.
  
\item \textbf{Output:}
  \begin{itemize}
  \item An instance $I_G$ created at Step~\ref{step:ig} 
  \item A coreset consisting of $\seeded \cup \coreset_{i,j}$, with offset $\offset$ defined at Step~\ref{step:offset} and weights for $\seeded$ defined throughout the algorithm, weights for $\coreset_{i,j}$ defined by the sampling procedure.
  \end{itemize}
\end{enumerate}

If the condition $\cost(\seeded) \leq \eps\cost(\constantApprox)/c_\constantApprox$ is met, then we say that $\seeded$ is
\emph{low-cost}; otherwise, we say it is \emph{locally stable}.

We will need to compare cost of solution in the original instance $I$, and in the one computed by the algorithm: $\cost$ will denote the cost in $I$ while 
 $\cost_{I_G}$ the one for instance $I_G$.

\subsection{Proof of the Greedy Seeding}
The outcome of the greedy seeding step, $I_G$, satisfies the following lemma:
\begin{lemma}
\label{lem:preprocess}
Suppose $\seeded$ is locally stable.
Let $\offset = \sum\limits_{C \in \seeded} \sum\limits_{p \in \out(C)} \cost(p, \seeded)$ be the cost of points in outer rings.
For every solution $\calS$, it holds that 
$$|\cost(\calS) - (\cost_{I_G}(\calS) + \offset)| \leq \eps \cost(\calS).\qedhere$$
\end{lemma}

Before proving \cref{lem:preprocess}, we show that the solution $\seeded$ has the
following properties:
\begin{lemma}
\mbox{}\label{lem:greedy} 
  $\seeded$ contains at most $O(k \cdot z\cdot \log(1/\eps)/\eps)$
  centers and, in the case where $\seeded$ is locally stable,
for any solution $\calS$ and any subset $Q \subseteq P$,
    $ \cost(Q, \seeded) \le  \cost(Q, \calS) + \eps \opt$.   
\end{lemma}
\begin{proof}
  Note that the cost of $\seeded$ is clearly at most the cost of $\constantApprox$, whose cost is at most $c_\constantApprox$ times the cost of
  the best solution for the $k$-clustering problem.
  
  Each center added by the greedy step decreases its cost
  by a factor $(1-\nicefrac{\eps}{c_\constantApprox k})$, and the initial solution $\seeded$ has cost $\cost(\seeded)$: since the algorithm stops if the cost
  drops below $\eps\cost(\constantApprox) / c_\constantApprox$, the total number of centers added is at most $c_\constantApprox k \log(1/\eps)/\eps$ The total number of centers is therefore at most $c_\constantApprox k \log(1/\eps)/\eps$.

  Note also that $\cost(\constantApprox) / c_\constantApprox \leq \opt$: hence, the inequality $(1-\nicefrac{\eps}{c_\constantApprox k})\cost(\seeded) \le \cost(\seeded \cup \{c\})$ implies 
  $\cost(\seeded) \le \cost(\seeded \cup \{c\}) + \eps \opt / k$. Hence,  
  since $\seeded$ is locally stable, there is no candidate center $c$ such that $\cost(\seeded) - \cost(\seeded \cup \{c\}) \ge \eps \opt/k$. Let $\calS$ be a solution, $Q$ be a subset of clients, $s$ be a center of 
the solution $\calS$ and $C_s$ be its cluster. Since the improvement made by adding $s$ is at least $\sum_{p\in C_s\cap Q} \cost(p, \seeded) - \cost(p, \calS)$, it holds that $\sum_{p\in C_s\cap Q} \cost(p, \seeded) - \cost(p, \calS) \le \cost(\seeded) - \cost(\seeded \cup \{c\}) \le \eps\opt/k$.
Summing over the $k$ centers of the solution $\calS$ yields the inequality $ \cost(Q, \calS) \le  \cost(Q, \seeded) + \eps \opt$, concluding the lemma.
\end{proof}

We can now turn to the proof of  \cref{lem:preprocess}, that shows that the the outer rings of solution $\seeded$ can simply be discarded, and their cost added in the offset $\offset$.

\begin{proof}[proof of \cref{lem:preprocess}]
 Fix a cluster $C$ of $\seeded$, with center $c$. 
 We first want to show that $\cost(C, \calS) \leq \cost_{I_G}(C, \calS)+ \cost(\out(C), \seeded) + \eps\cost(C, \calS)$. First, this is equivalent to
 $\cost(\out(C), \calS) \leq |\out(C)|\cost(c, \calS)+ \cost(\out(C), \seeded) + \eps\cost(C, \calS)$, as all points in $C\setminus \out(C)$ have same cost in the original instance and in $I_G$.
 
 Using \cref{lem:weaktri}, we have 
$$\cost(\out(C), \calS) \leq (1+\eps)\cost(\out(C), \seeded) + (1+z/\eps)^{z-1} |\out(C)|\cost(c, \calS).$$
 Hence, one needs to bound $|\out(C)|\cost(c, \calS)$. For that, we show that the cost of clients in $\out(C)$ can be charged to clients of $\inner(C) \cup \main(C)$. 
 
 First note by Markov's inequality, $|\out(C)| \leq \left(\nicefrac \eps z\right)^{2z}|C|$. 
  Hence, one can partition $\inner(C) \cup \main(C)$ into parts of size at least $s = (|C|-|\out(C)|)/|\out(C)|$, and assign every part to a point in $\out(C)$ in a one-to-one correspondence. 

   For such a point $p \in \out(C)$ consider the $s$
  points $p_1,\ldots,p_s$ of the part assigned to it. We show how to charge $\cost(c, \calS)$ to the cost of those points. Let $\calS(p_j)$ be the center serving $p_j$ in the solution $\calS$. It holds that $\cost(c, \calS) \leq \min_j \cost(c, \calS(p_j))$. Averaging, this is at most $\frac{1}{s} \sum_{j=1}^{s}
  \cost(c, \calS(p_j))$ which is due to \cref{lem:weaktri}
  at most $\frac{1}{s}\sum_{j=1}^{s}  (1+ \eps) \cost(p_j, \seeded) + (1+z/\eps)^{z-1}\cost(p_j, \calS)$. Since $1/s \leq 2|\out(C)| / |C| = 2\left(\nicefrac \eps z\right)^{2z}$, 
  we conclude (using $2\left(\frac{\eps}{z}\right)(1+\eps) \leq 1$) that:

\begin{align}
\notag
 |\out(C)|\cost(c, \calS)
 &\leq 2\left(\frac{\eps}{z}\right)^{z+1} (1+\eps) \sum_{p \in \inner(C) \cup \main(C)} \cost(p, \seeded) + \cost(p, \calS)\\
&\leq \left(\frac{\eps}{z}\right)^{z}(\cost(C, \seeded) + \cost(C, \calS)).
\label{eq:out-1}
\end{align}
  
  \bigskip
We apply this inequality to bound $\cost(\out(C), \calS)$:
\begin{align*}
\cost(\out(C), \calS) 
&\leq (1+\eps)\cost(\out(C), \seeded) + (1+z/\eps)^{z-1} |\out(C)|\cost(c, \calS)\\
&\leq \cost(\out(C), \seeded) + 2\eps(\cost(C, \seeded) + \cost(C, \calS))
\end{align*}  
  Hence, summed over all cluster, this yields
  \begin{equation}
  \label{eq:out-2}
  \cost(\out, \calS) \leq \cost(\out, \seeded) + 2\eps(\cost(\seeded) + \cost(\calS))
  \end{equation}
  
We now turn to the other direction of the inequality. Using \cref{lem:greedy}, 
  \begin{align*}
  \cost(\out, \seeded) &\leq \cost(\out, \calS) + \eps \opt\\
  &\leq \cost(\out, \calS) + \eps \cost(\calS) 
  \end{align*}

  Combined with \cref{eq:out-1}, this yields
  \begin{align}
  \cost(\out, \seeded) + \cost_{I_G}(\out, \calS) - \cost(\out, \calS) \leq  2\eps (\cost(\calS) + \cost(\seeded)).\label{eq:out-3}
  \end{align}
  
  Combining \cref{eq:out-2} and \cref{eq:out-3}, and using that all points not in $\out$ have same cost in the original instance and in $I_G$ concludes the lemma.
\end{proof}

\vspace{0.5em}
\subsection{Computing Coresets via $\eps$-set-approximation: Proof of \cref{lem:approxCoreset}}
\vspace{0.5em}

We first show that an $\eps$-set-approximation of the set system $(R_{i,j}, \calB)$ is indeed a coreset for $R_{i,j}$. We will then bound the size of such an approximation. In this section, we fix the cluster $C_i$ and a ring $R_{i,j}$.

\begin{lemma}\label{lem:approxCoreset}
Let $R_{i,j}$ be a ring, and $\Omega$ be an $\eps'$-set-approximation for $(R_{i,j}, \calB)$ with $\eps' = 20\frac{8^z \eps^2}{\log(4z/\eps)}$. Then $\Omega$ with uniform weights $|R_{i,j}|/|\Omega|$ is an $\eps$-coreset for $R_{i,j}$.
\end{lemma}

We use the same proof technique as \cite{Cohen-AddadSS21} to analyze the algorithm. More precisely, we define ranges of clients: we let $I_{i,j,\ell}$ to be the set of points from $R_{i,j}$ that pays roughly $(1+\eps')^\ell$ is $\calS$, i.e., $I_{i,j,\ell} := \{p \in R: (1+\eps'')^{\ell} \leq \cost(p, \calS) < (1+\eps'')^{\ell + 1}\}$, where $\eps''$ will be set later. 

This analysis distinguishes between three cases:
\begin{enumerate}
\item $\ell \leq j + \log_{1+\eps''} \eps $, in which case we say that $I_{\ell}$ is \emph{tiny}. The union of all tiny ranges is denoted $I_{tiny, \calS}$.
\item  $ j + \log_{1+\eps''} \eps \leq \ell \leq j + \log_{1+\eps''} (4z/\eps)$, in which case we say $I_{\ell}$ is \emph{interesting}. 
\item $\ell \geq  j +\log_{1+\eps''}(4z/\eps)$, in which case we say $I_{\ell}$ is \emph{huge}.
\end{enumerate}

The very same proof as Lemma 5 and 7 of \cite{Cohen-AddadSS21} shows that the cost the tiny and huge ranges are preserved. For completeness, we provide statement and proofs in \cref{app:proofCoreset}. For the interesting ranges, we use properties of $\eps'$-set-approximation as follows.

\begin{lemma}
Let $R_{i,j}$ be a ring, and $\calS$ be a solution such that all huge $I_{\ell}$ are empty. Further, let  $\Omega$ be an $\eps'$-set-approximation for $(R, \calB)$, and uniform weights $|R|/|\Omega|$ for point in $\Omega$.  Then, if $\eps'' = \eps / 10$ and $\eps' = 20 \frac{2^{3z} \eps^2}{\log(4z/\eps)}$, it holds that $\cost(\Omega, S) = (1\pm \eps) \cost(R_{i,j}, S)$.
\end{lemma}
\begin{proof}
Fix some solution $\calS$. For simplicity, we drop the subscript $i, j$ and let $R := R_{i,j}, I_\ell := I_{i,j,\ell}$. 

To show that $\cost(\Omega, \calS) = (1\pm \eps) \cost(\calS)$, we will show the following stronger property: for any $\ell$, $$|\cost(I_{\ell}, \calS) - \cost(I_\ell \cap \Omega, \calS)| \leq \frac{\eps}{\log(4z/\eps)} \cost(I_\ell, \calS).$$

 Hence, we fix some integer $\ell$. Note that $I_\ell = H_{\calS, (1+\eps'')^{\ell+1}} \setminus H_{\calS, (1+\eps'')^{\ell}}$, where $H_{\calS, r} = \{p \in P : \dist(p, \calS) \geq r\} \in \calB$.
 Since $\Omega$ is an $\eps'$-set-approximation for $(P, \calB)$, it holds that 
 \begin{align*}
     \left| \frac{|H_{\calS, (1+\eps'')^{\ell+1}}|}{|R|} - \frac{|H_{\calS, (1+\eps'')^{\ell+1}} \cap \Omega|}{|\Omega|}\right| \leq \eps'\\
     \left| \frac{|H_{\calS, (1+\eps'')^{\ell}}|}{|R|} - \frac{|H_{\calS, (1+\eps'')^{\ell}} \cap \Omega|}{|\Omega|}\right| \leq \eps'
 \end{align*}
 Combining those two guarantees with triangle inequality ensure that
 \begin{align*}
     \left| \frac{|I_\ell|}{|R|} - \frac{|I_\ell \cap \Omega|}{|\Omega|}\right| &\leq \left| \frac{|H_{\calS, (1+\eps'')^{\ell+1}}|}{|R|} - \frac{|H_{\calS, (1+\eps'')^{\ell+1}} \cap \Omega|}{|\Omega|}\right| + \left| \frac{|H_{\calS, (1+\eps'')^{\ell}}|}{|R|} - \frac{|H_{\calS, (1+\eps'')^{\ell}} \cap \Omega|}{|\Omega|}\right|\\
     &\leq 2\eps'.
 \end{align*}

  Furthermore, by definition of $I_\ell$, $|I_\ell| \cdot (1+\eps'')^\ell \leq \cost(I_\ell, \calS) < |I_\ell| \cdot (1+\eps'')^{\ell+1}$. Similarly, due to the weighting of points in $\coreset$, we have $(1+\eps'')^{\ell} \frac{|R| \cdot |I_\ell \cap R|}{|\coreset|} \leq \cost(I_\ell \cap \coreset, \calS)$ Combining those equations, we get
 
 \begin{align*}
     \cost(I_\ell, \calS) &\leq |I_\ell| \cdot (1+\eps'')^{\ell+1} \\
     &\leq |I_\ell \cap \Omega|\cdot(1+\eps'')^{\ell+1} \cdot \frac{|R|}{|\Omega|} + 2\eps' \cdot |R| (1+\eps'')^{\ell+1}\\
     &\leq (1+\eps'') \cost(I_\ell \cap \Omega, \calS) + 2\eps' \cdot |R| (1+\eps'')^{\ell+1}.
 \end{align*}   
 
 To bound further the right hand side, we let $q \in I_\ell$, so that $(1+\eps'')^\ell \leq \cost(q, \calS)$:
 \begin{align*}
     |R|(1+\eps'')^{\ell} &\leq \sum_{p \in R} \cost(q, \calS) 
     \leq \sum_{p \in R} 2^{z-1}\cost(p, q) + 2^{z-1} \cost(p, \calS)\\
     &\leq 2^{z-1} \cost(R, \calS) + 2^{z-1} \sum_{p \in R} 2^{z-1} \left(\cost(p, \seeded) + \cost(q, \seeded)\right)\\
     &\leq 2^{z} \cost(R, \calS) + 2^{3z-1} \cost(R, \seeded),
 \end{align*}
 where the last line uses $\cost(q, \seeded) \leq 2^z \cost(p, \seeded)$, since $p$ and $q$ are in the same ring.
 
Combined with the previous inequality, this yields
$\cost(I_\ell, \calS) \leq (1+\eps'') \cost(I_\ell \cap \Omega, \calS) + 2^{3z}\eps' \cdot \left(\cost(R, \calS) + \cost(R, \seeded)\right)$.

 Similarly, we get
 \begin{align*}
     \cost(I_\ell \cap \Omega, \calS) \leq (1+\eps'') \cost(I_\ell, \calS) + 2^{3z} \eps' (\cost(R, \calS) + \cost(R, \seeded)).
 \end{align*}

Summing that over the $\log_{1+\eps''}(4z/\eps)$ interesting ranges and combining with \cref{lem:kepstiny} gives that:
\begin{align*}
&|\cost(I_\ell, \calS) - \cost(I_\ell \cap \Omega, \calS)| \\
\leq ~ & \log_{1+\eps''}(4z/\eps) \cdot 2^{3z}\eps' \cdot (\cost(R, \calS) + \cost(R, \seeded))+ \eps'' \cost(R \cap \Omega, \calS) + \eps'' \cost(R, \calS).
\end{align*}

Hence, we have:
 \begin{eqnarray*}
     &|\cost(I_\ell, \calS) - \cost(I_\ell \cap \Omega, \calS)| &\leq \log_{1+\eps''}(4z/\eps) \cdot 2^{3z}\eps' \cdot (\cost(R, \calS) + \cost(R, \seeded)) \\
     & & ~~ + \eps'' \cost(R \cap \Omega, \calS) + \eps'' \cost(R, \calS)\\
     & &\leq \log_{1+\eps''}(4z/\eps) \cdot 2^{3z}\eps' \cdot (\cost(R, \calS) + \cost(R, \seeded)) \\
     & & ~~+ 2 \eps'' \cost(R, \calS) + \eps''|\cost(R, \calS) - \cost(R \cap \Omega, \calS)|\\
     & &\leq \left(\log_{1+\eps''}(4z/\eps) \cdot 2^{3z}\eps' + 2\eps''\right) \cdot (\cost(R, \calS) + \cost(R, \seeded))\\
     & & ~~+ \eps''|\cost(R, \calS) - \cost(R \cap \Omega, \calS)|,
 \end{eqnarray*}
which is equivalent to
\begin{equation*}
 |\cost(R, \calS) - \cost(R\cap \Omega, \calS)| \leq \frac{\log_{1+\eps''}(4z/\eps) \cdot 2^{3z}\eps' + 2\eps''}{(1-\eps'')}(\cost(R, \calS) + \cost(R, \seeded).
\end{equation*}
Chosing $\eps''$ such that $\frac{2\eps''}{1-\eps''} \leq \frac{\eps}{2}$ (e.g. $\eps'' = \eps / 10$) and $\eps'$ such that $\frac{\log_{1+\eps''}(4z/\eps) \cdot 2^{3z}\eps'}{1-\eps''} \leq \frac{\eps}{2}$ (e.g. $\eps' = 20 \frac{2^{3z} \eps^2}{\log(4z/\eps)}$) concludes the lemma.
\end{proof}

Combining the results for tiny, interesting and huge ranges concludes the proof of \cref{lem:approxCoreset} and shows that $\coreset_{i,j}$ is an $\eps$-coreset of $R_{i,j}$. Combined with \cref{lem:preprocess}, this concludes the coreset guarantee of \cref{thm:smallcoreset}. 
The size of the coreset for a given ring is $O\left(2^{O(z\log z)}D \eps^{-4} \log(D/\eps)\right)$ from \cref{thm:chazelle}, and there are \\$O\left(2^{O(z\log z)} k\eps^{-1} \polylog(1/\eps)\right)$ many rings. Hence, the total size of the coreset constructed is $O\left(2^{O(z\log z)}\cdot kD \eps^{-k} \log(D) \polylog(1/\eps)\right)$
It only remains to bound the time complexity of the algorithm.

\subsection{Complexity Analysis}
Two steps dominate the running time of the algorithm: the greedy seeding, and the computation of the $\eps'$-set-approximation. For the latter, a direct combination of \cref{thm:chazelle} shows that all the $\eps'$-set-approximation can be computed in time $k (D/\eps)^{-O(d)} \cdot |X|$.

Hence, it only remains to bound the running time of the greedy seeding. 

In the discrete metric case, this is straightforward: each iteration costs $O(|X|\cdot |P|)$, and there are $\tilde O(k/\eps)$ many of them (see \cref{lem:greedy}).

The Euclidean case requires more work, as it is not possible to enumerate over all candidate center. To implement the greedy seeding, we first use \cref{thm:dimred} to reduce the dimension to $\eps^{-O(z)} \log k$. In this low-dimensional space, \cref{lem:bicriteria-lowdim} shows how to implement the greedy seeding algorithm in time $n \log n k^{\eps^{-O(z)}}$.

If the solution computed is low cost in the projected space, it is as well low-cost in the original space. Furthermore,
\cref{thm:dimred} ensures that any locally stable solution in the projected space is also locally stable for the original input, therefore our argument follows. Note that it would not be possible to compute an arbitrary bicriteria solution, for instance using \cref{thm:bicriteria}: this would not ensure the local stability. Hence, we crucially need to apply the dimension reduction from \cref{thm:dimred}, and not a terminal embedding or a simpler dimension reduction. This concludes the proof of \cref{thm:smallcoreset}.

%% file: detFPT.tex
\section{Deterministic Bicriteria Approximation} \label{sec:bicriteria}

The goal of this section is to prove \cref{thm:bicriteria}, that we recall for convenience:
\bicriteria*

In order to compute a good solution efficiently, we rely on two results allowing to reduce the dimension.
We use a result from Makarychev, Makarychev and Razenstein shows that their exist a dimension reduction onto $O\lpar \log(k/\eps)/\eps^2\rpar$ that preserves the cost of any clustering. A theorem from Kane and Nelson \cite{KaneN14} allows to generate this projection using few random bits: hence, we can enumerate over all realization of those bits, compute the best solution in the projected space and lift it back to the original one. 

In the low dimensional space, we can compute an approximate solution using a standard greedy procedure. For completeness, we provide proofs in \cref{sec:bicrit-lowdim}. Before that, we state those two key theorems.

\begin{theorem}\label{thm:subGaussianKN}[Theorem 15 in Kane and Nelson]
For any integer $d$ and any $0<\eps < 1/2$, there exists an probability distribution on $d \times m$ random matrices such that for any $x \in \R^d$ with $\|x\| = 1$,
\[\mathbb{P}_\Pi \left[\left|\|\Pi x\|^2 - 1\right| \geq \eps\right] \leq \exp(-\eps^2 m).\] 
Moreover, $\Pi$ can be sampled using $O\lpar \eps^2 m  \cdot \log d \rpar$ many random bits.

Given a seed and a vector $x \in \R^d$, the embedding can be performed in polynomial time.
\end{theorem}

\begin{theorem}\label{thm:dimRedMMR}[\cite{MakarychevMR19}]
The family of random maps defined in \cref{thm:subGaussianKN} verifies, for $m = O(\log (k/\eps)/\eps^2)$ the following: for any set $P \subseteq \R^d$ and any partition $\calC$ of $P$ into $k$ parts,
\[\mathbb{P}_\Pi[\costP\lpar\Pi(P), \Pi(\calC)\rpar = (1\pm \eps) \costP(P, \calC)] \geq 1/2\]
\end{theorem}

\subsection{Computing a Bicriteria Approximation in Low Dimensional Spaces}\label{sec:bicrit-lowdim}

In this section, we show the following lemma:

\begin{lemma}\label{lem:bicriteria-lowdim}
There exists an algorithm running in time $n \log n (1/\eps)^{O(d)}$ that computes a set of $O(k\cdot \log(1/\eps)/\eps)$ many centers $\calS$, such that $\cost(P, \calS) \leq (1+\eps)\opt$.
\end{lemma}

The algorithm is the standard bicriteria approximation. We provide proofs for completeness: our contribution is presented in \cref{sec:highToLow}.

\begin{itemize}
    \item Start computing an $\alpha$-approximation $S_0$ to $(k, z)$-clustering on $A$. \item As long as there exists a center $c$ such that $\cost(S_i \cup \{c\} \leq (1-\frac{\eps}{\alpha k}) \cost(S_i)$, do $S_{i+1} \leftarrow S_i \cup \{c\}$ and $i \leftarrow i+1$.
    \item Stop this procedure when there is no such center $c$, or $\cost(S_i) \leq \frac{\eps}{\alpha}\cost(S_0)$, and let $S$ be the final solution.
\end{itemize} 

We break the proof of \cref{lem:bicriteria-lowdim} into two parts: first we show the quality of the solution, and then how to implement fast the algorithm. 

\begin{lemma}\label{lem:greedy-bicriteria}
$\calS$ is a $\lpar(1+\eps), \alpha \log(1/\eps)/\eps + 1\rpar$-bicriteria approximation for $(k, z)$-clustering on $A$.
\end{lemma}
\begin{proof} 
  Each center added decreases the cost
  by a factor $(1-\nicefrac{\eps}{\alpha k})$, and the initial solution $S_0$ has cost $\cost(S_0)$: since the algorithm stops if the cost
  drops below $\eps\cost(S_0) / \alpha$, the total number of centers added is at most $\alpha k \log(1/\eps)/\eps$. The total number of centers is therefore at most $\alpha k \log(1/\eps)/\eps + k$.

  In the case where the procedure stops because $\cost(S) \leq \frac{\eps}{\alpha}\cost(S_0)$, then $S$ has cost at most $\eps$ times the optimal cost, since $S_0$ is an $\alpha$-approximation. In that case, we are done.
  
  In the other case, we note that $\cost(S_0) / \alpha \leq \opt$: hence, the inequality $(1-\nicefrac{\eps}{\alpha k})\cost(S) \ge \cost(S \cup \{c\})$ implies 
  $\cost(S) \ge \cost(S \cup \{c\}) + \eps \opt / k$. Hence,  there is no candidate center $c$ such that $\cost(S) - \cost(S \cup \{c\}) \ge \eps \opt/k$. 
Summing over the $k$ centers of the optimal solution $\text{OPT}$ yields the inequality $\cost(S) - \cost(S \cup \text{OPT}) \ge \eps \opt$, hence $S$ is a $(1+\eps)$-approximation of 
OPT. This concludes the lemma.
\end{proof}

It remains to show how to implement that greedy procedure fast. For that, we discretize the set of candidate centers, to be able to check all of them.

For a point $p \in P$ and a radius $r$, we let $\cand_{p, r}$ be an $\frac{\eps}{z} \cdot r$-cover of the ball $B(p, r)$. 
Given some $\alpha$-approximation, we let $\Delta$ be the average cost of that solution.

We define $\cand := \cup_{p \in P} \cup_{i = \log (\eps / (\alpha z))}^{\log (n/\alpha)} \cand_{p, 2^{i/z} \Delta^{1/z}}$. We show that their exists a near optimal solution with centers in $\cand$:

\begin{lemma}
There exists a solution $\calS \in \cand^k$ with cost $(1+3\eps) \opt$. Furthermore, $|\cand| \leq n \log n \cdot \eps^{-O(d)}$ and can be computed in $O(|\cand|)$ time.
\end{lemma}
\begin{proof}
First, the size of $\cand$ follows immediately from the size of covers in $\R^d$: there exists $\eps$-covers of $B(0,1)$ of size $\eps^{-O(d)}$ that can be computed in time $\eps^{-O(d)}$.

Given a cluster $C$ of the optimal solution, let $c^*$ be the optimal center for $C$. We show there exist a good approximation $c$ to $c^*$ in $\cand$.

For this, let $p_0$ be the closest point in $C$ to $c^*$, and $i$ such that $2^{(i-1)/z} \Delta^{1/z} \leq \dist(p_0, c^*) \leq 2^{i/z} \Delta^{1/z}$.

First, if $i \leq \log(\eps/(\alpha z))$, we let $c := p_0$ so that $\dist(c, c^*) \leq \frac{\eps}{2\alpha z} \Delta$.

Hence, for any point $p \in C$, we have:
\begin{align*}
&|\cost(p, c) - \cost(p, c^*)| \\ 
&\leq \eps \cost(p, c^*) + \lpar \frac{\eps + z}{\eps}\rpar ^{z-1} \cost(c, c^*)\\
&\leq \eps \cost(p, c^*) + \lpar \frac{\eps + z}{\eps}\rpar ^{z-1} \lpar \frac{\eps}{2\alpha z}\rpar ^z \Delta\\
&\leq \eps \cost(p, c^*) + \eps \Delta.
\end{align*}

In the other case, it must be that $i-1 \leq \log (n/\alpha)$: otherwise, $\opt \geq \cost(p, c^*) \geq 2^{i-1} \Delta \geq n \Delta / \alpha$, and by definition of $\Delta$ we knew $n \Delta \leq \alpha \opt$, a contradiction.

Let $c$ be the closest point to $c^*$ in $\cand_{p_0, i}$: by the previous discussion, $c \in \cand$ and the definition of $\cand_{p_0, i}$ ensures that $\dist(c^*, c) \leq \frac{\eps}{4z} 2^i \leq \frac{\eps}{2z}\dist(p_0, c^*)$. 

Hence, for any point $p \in C$, we have just as before:
\begin{align*}
&|\cost(p, c) - \cost(p, c^*)| \\
&\leq \eps \cost(p, c^*) + \lpar \frac{\eps + z}{\eps}\rpar ^{z-1} \cost(c, c^*)\\
&\leq \eps \cost(p, c^*) + \lpar \frac{\eps + z}{\eps}\rpar ^{z-1} \lpar \frac{\eps}{2z}\rpar ^z \cost(p_0, c^*)\\
&\leq 2\eps \cost(p, c^*).
\end{align*}

Summing over all points of $C$ gives $\cost(C, c) \leq (1+2\eps) \cost(C, c^*) + \eps/\alpha \cdot |C| \Delta$. Summing now over all clusters and using $n \Delta \leq \alpha \opt$ ensures that there is a solution with cost at most $(1+3\alpha) \opt$, with center in $\cand$.
\end{proof}

Hence, there exists a $(1+\eps)$-approximation in $\mathbb{C}^k$. Using this solution in place of the optimal one in the final equation of Lemma~\ref{lem:greedy-bicriteria} ensures that the solution $\calS$ given by the greedy algorithm restricted to pick centers in $\mathbb{C}$ yields a $(1+\eps)$-approximation.

 Using local-search (see e.g. Theorem 3.2 in~\cite{GT08}) for computing the constant-factor approximation gives $\alpha = z^{O(z)}$: the solution $\calS$ is thus a $(1+\eps, z^{O(z)} \log(1/\eps)/\eps)$-bicriteria approximation. For the running time, we use the fact that given one candidate center, verifying whether adding it to the current solution decreases enough the cost takes time $O(nd)$. Hence, the running time of this greedy bicriteria algorithm is $n \log n \cdot \eps^{-O(d)} \cdot nd \log(1/\eps) / \eps $, which concludes the proof of \cref{lem:bicriteria-lowdim}.

\subsection{From High to Low dimension, and Vice-versa}\label{sec:highToLow}

Our first step is to reduce the dimension, in order to reduce the number of random bits required by \cref{thm:subGaussianKN}.

For that, we use a result from \cite{MakarychevMR19}: preserving pairwise distances between the input points is enough to preserve the cost of any cluster. The proof is based on Kirszbraun theorem, to find a Lipschitz extension of the function mapping the projections of the input set to the original set. 
Using \cref{thm:EIO}, such a projection can be computed in deterministic polynomial time.
It is therefore enough to focus on a space of dimension $\tilde d = O\lpar \eps^{-2} \log n\rpar$.

Assuming now we are provided a $\Pi : \R^{\tilde d} \rightarrow \R^m$ that verifies \cref{thm:dimRedMMR}, computing a good solution for $P$ boils down to computing one for $\Pi(P)$ in $\R^m$, with $m = O\lpar \eps^{-2}\log(k/\eps)\rpar$.
The algorithm from \cref{lem:bicriteria-lowdim} gives a clustering with near-optimal cost in the low dimensional space: to lift it in the original space, one can simply compute the optimal center in $\R^{d}$ for each cluster. This can be done using convex programming in polynomial time.

To apply this idea, we enumerate over all possible seeds of bits of length
\[O\lpar \log \tilde d \cdot \log(\eps^{-1} k)  \rpar = O\lpar \log (\eps^{-2} \log n) \log(\eps^{-1} k)\rpar.\]
 \cref{thm:subGaussianKN} ensures that one of them will verify \cref{thm:dimRedMMR}. Our algorithm is the following: given a seed, compute a near-optimal clustering in the projected space. Then, compute the cost of that clustering in the original space $\R^d$. Output the clustering with minimal cost among those. Since at least one projection preserves the cost of all clustering, this ensures that the minimal cost is $(1+O(\eps))\opt$.

The running time is 
\begin{itemize}
\item $2^{O\lpar \log (\log (n) / \eps) \cdot \log(k/\eps)\rpar} = k^{O\lpar \log(1/\eps)\rpar} \cdot \log(n)^{O(\log(1/\eps))} \cdot 2^{O\lpar \log \log n \cdot \log k \rpar} \leq  k^{O(\log(1/\eps) + \log \log k)} \cdot n^{1+o(1)}$ to enumerate over all seeds, where in the last inequality we used either $\log k \leq \log(n)/\log \log n$, then the last term is at most $n$, or $\log k > \log(n)/\log \log n$, then $\log \log n \leq 2\log \log k$ and the last term is at most $k^{\log \log k}$,
\item for each seed, polynomial time to project onto a $\log(k/\eps)/\eps^2$ dimensional space, $n^2 d \log n (1/\eps)^{O(\log (k/\eps) / \eps^{-2}}$  to compute a $(1+\eps)$-approximate solution with $k\log(1/\eps)$ many centers in that space, and polynomial time for evaluating the cost of the clustering in $\R^d$
\end{itemize}
This concludes the proof of \cref{thm:bicriteria}.

%% file: appendix.tex
\appendix

\section{Witness Sets}
\label{sec:witness}

For proving Theorem~\ref{thm:witness}, we rely on the existence of small coreset for $(1, z)$-clustering:
\begin{lemma}\label{lem:coreset-huang}[See Cohen-Addad, Saulpic and Schwiegelshohn~\cite{neurips21}]
Let $A$ be a set of $n$ points in $\mathbb{R}^d$. There exists a set $\Omega_A$  of size $\tilde O(\eps^{-2}2^O(z))$ and weights $w : \Omega_A \rightarrow \mathbb{R}_+$ such that, for every possible center $c$, 
\[\left\vert \sum_{i=1}^n\|A_i-c\|^z - \sum_{i\in \Omega}w(i)\|A-c\|^z \right\vert \leq \eps \sum_{i=1}^n\|A_i-c\|^z.\]
\end{lemma}

{\rm
We use those coresets as follow. Let $a \in A$ be a center inducing a $2^z$-approximation for $(1,z)$-clustering on $A$, and let $\Delta_a$ be the average cost for that solution (note that $\Delta = \Delta=\frac{\sum_{i=1}^n\|A-c\|^z}{n} \leq 2^z \Delta_a$). We show how to use Lemma~\ref{lem:coreset-huang} to show that their exists a coreset $\Omega_A$ where all points are at distance at most $\frac{z}{\eps} \cdot \left(\Delta_a\right)^{1/z}$ of $a$. This is enough to conclude the existence of $(\frac{2z}{\eps}, \eps^{-O(1)}\cdot 2^{O(z)} ,\eps)$-uniform witness set: the optimal center for $\Omega_A$ must lie in the convex hull of $\Omega_A$, and the diameter of $\Omega_A$ is at most $\frac{2z}{\eps} \Delta$.
}

\begin{lemma}
Let $A$ be a set of $n$ points in $\mathbb{R}^d$, and $a \in A$ be a center inducing a $2^z$-approximation for $(1,z)$-clustering on $A$. 

There exist a set $\Omega_A$  of size $\eps^{-2} 2^{O(z)}$ and weights $w : \Omega_A \rightarrow \mathbb{R}_+$ such that: every point of $\Omega_A$ is at distance at most $\frac{2z}{\eps} \Delta$ of $a$, and for every possible center $c$, 
\[\left\vert \sum_{i=1}^n\|A_i-c\|^z - \sum_{i\in \Omega}w(i)\|A-c\|^z \right\vert \leq \eps \sum_{i=1}^n\|A_i-c\|^z.\]
\end{lemma}
\begin{proof}
Let $O_A$ be the set of points at distance more than $\left(\frac{z}{\eps}\right)^{2} \Delta_a^{1/z}$ of $a$. Let $p$ be any point at distance 
$\left(\frac{z}{\eps}\right)^{2} \Delta_a^{1/z}$ of $a$\footnote{this distance is somewhat arbitrarily picked in order to simplify the proof.}.
 We start by showing that, for all center $c$,
\begin{equation}
\left\vert \cost(A, c) - \left( \cost(A \setminus O_A, c) + \frac{\cost(O_A, a)}{\cost(p, a)}\cost(p, c)\right)\right\vert \leq O(\eps)\left(\cost(A, c) + \cost(A, a)\right). \label{eq:far}
\end{equation}

Proving this inequality would conclude the lemma. Indeed, since $2^{-z}\cost(A, a) \leq \cost(A, c)$, the right hand side can be upper-bounded by $O(\eps 2^z) \cost(A, c)$. Hence, adding the point $p$ with weight $w(p) := \frac{\cost(O_A, a)}{\cost(p, a)}$ to an $\eps$-coreset $\Omega^1_A$ (with weights $w$) of $A \setminus O_A$ provided by  Lemma~\ref{lem:coreset-huang} gives a $O(\eps 2^z)$-coreset for $A$. Rescaling $\eps$ by $2^z$ concludes.

To show Equation~\ref{eq:far}, we fix a center $c$. We start by bounding $|O_A|\cost(a, c)$. For that, we show that the cost of clients in $O_A$ can be charged to clients of $A \setminus O_A$.  First note that by averaging, $|O_A| \leq \left(\nicefrac \eps z\right)^{2z}|A|$ -- otherwise $\cost(O_A, a) \geq |O_A|\left(\frac{z}{\eps}\right)^{2z} \Delta_a > |A| \Delta_a = \cost(A, a)$. 
  Hence, one can partition $A \setminus O_A$ into parts of size at least $s = (|A|-|O_A|)/|O_A|$, and assign every part to a point in $O_A$ in a one-to-one correspondence. 

   For such a point $p \in O_A$ consider the $s$
  points $p_1,\ldots,p_s$ of the part assigned to it. We show how to charge $\cost(a, c)$ to the cost of those points. By the modified triangle inequality, for all $i$ $\cost(a, c) \leq (1+\eps)\cost(a, p_i) + (1+z/\eps)^{z-1} \cost(p_i, c)$. Taking an average over all $i$ gives that $\cost(a, c)$ is  
  at most $\frac{1}{s}\sum_{j=1}^{s}   \eps \cost(p_j, a) + (1+z/\eps)^{z-1}\cost(p_j, c)$. Since $1/s \leq 2|O_A| / |A| = 2\left(\nicefrac \eps z\right)^{2z}$, 
  we conclude that (using $2\left(\frac{\eps}{z}\right)(1+\eps) \leq 1$)

\begin{align}
\notag
|O_A|\cost(a, c) &\leq 2\left(\frac{\eps}{z}\right)^{z+1} (1+\eps) \sum_{p \in A \setminus O_A} \cost(p, a) + \cost(p, c)\\
&\leq \left(\frac{\eps}{z}\right)^{z}(\cost(A, a) + \cost(A, c))
\label{eq:out-1-app}
\end{align}

We now show how to use this inequality to show Equation~\ref{eq:far}, and start by showing \\
$\frac{\cost(O_A, a)}{\cost(p, a)}\cost(p, c) \leq (1+O(\eps))\cost(O_A, c) + O(\eps) \cost(A, c)$.

  We proceed as follows: we decompose the left-hand-side using Lemma~\ref{lem:weaktri}. Let $p$ be the point chosen to represent $O_A$:
  \begin{align*}
  \frac{\cost(O_A, a)}{\cost(p, a)} \cost(p, c)
  &\leq (1+\eps)\cost(O_A, a) + (1+z/\eps)^{z-1} \left(\frac{\cost(O_A, a)}{\cost(p, a)}\right)\cost(a, c)
  \end{align*}
  
  We bound separately the two terms. First,
  \begin{align*}
\cost(O_A, a) 
&\leq (1+\eps)\cost(O_A, c) + (1+z/\eps)^{z-1} |O_A|\cost(a, c)\\
&\leq (1+\eps)\cost(O_A, c) + (1+z/\eps)^{z-1} \left(\frac{\eps}{z}\right)^z(\cost(A, a) + \cost(A, c))\\
&\leq (1+\eps)\cost(O_A, c) + \eps (\cost(A, a) + \cost(A, c)\\
& \leq \cost(O_A, c) + O(\eps) (\cost(A, a) + \cost(A, c).
\end{align*}

Similarly, using $\cost(p, a) = \left(\frac{z}{\eps}\right)^{2z} \frac{\cost(A, a)}{|A|} \geq \left(\frac{z}{\eps}\right)^{2z} \frac{\cost(O_A, a)}{|A|}$:
  \begin{align}
  \notag
(1+z/\eps)^{z-1} \left(\frac{\cost(O_A, a)}{\cost(p, a)}\right)\cost(a, c)
&\leq  (1+z/\eps)^{z-1} \cdot \left(\frac{\eps}{z}\right)^{2z} |A| \cost(a, c)\\
\notag
&\leq \left(\frac{\eps}{z}\right)^{z} \left((1+\eps)\cost(A, c) + (1+z/\eps)^{z-1}\cost(A, a)\right)\\
&= O(\eps)(\cost(A, c) + \cost(A, a)). \label{eq:boundFar}
\end{align}

Hence,
\begin{align}
\frac{\cost(O_A, a)}{\cost(p, a)}\cost(p, c) &\leq \cost(O_A, c) + O(\eps)\left(\cost(A, c) + \cost(A, a)\right). \label{eq:costOA1}
\end{align}

We now turn to the other direction, namely 
$$\cost(O_A, c) \leq (1+\eps)\frac{\cost(O_A, a)}{\cost(p, a)}\cost(p, c) + O(\eps) \left(\cost(A, c) + \cost(A, a)\right).$$

We write
\begin{align*}
\cost(O_A, c) &\leq (1+\eps)\cost(O_A, a) + (1+z/\eps)^{z-1}|O_A|\cost(a, c),
\end{align*}
and bound here as well the two terms separately. First we have, using Equation~\ref{eq:boundFar}:
\begin{align*}
\cost(O_A, a) &= \frac{\cost(O_A, a)}{\cost(p, a)}\cost(p, a)\\
&\leq (1+\eps)\frac{\cost(O_A, a)}{\cost(p, a)}\cost(p, c) + \left(\frac{z}{\eps}\right)^{z-1} \frac{\cost(O_A, a)}{\cost(p, a)}\cost(a, c)\\
&\leq (1+\eps)\frac{\cost(O_A, a)}{\cost(p, a)}\cost(p, c) + O(\eps)\left(\cost(A, c) + \cost(A, a)\right)
\end{align*}

The second term is bounded by Equation~\ref{eq:out-1-app}: 
\[(1+z/\eps)^{z-1}|O_A|\cost(a, c) \leq \eps \left(\cost(A, c) + \cost(A, a)\right).\]

Hence,
\begin{gather}
    \cost(O_A, c) \leq (1+\eps)\frac{\cost(O_A, a)}{\cost(p, a)}\cost(p, c) + O(\eps) \cost(A, c). \label{eq:costOA2}
\end{gather}

Combining Equations \ref{eq:costOA1} and \ref{eq:costOA2} yields:
\begin{align*}
\left\vert \cost(O_A, c) - \frac{\cost(O_A, a)}{\cost(p, a)}\cost(p, c) \right \vert \leq O(\eps)\left(\frac{\cost(O_A, a)}{\cost(p, a)}\cost(p, c)  + \cost(A, c)\right).
\end{align*}
Using $\frac{\cost(O_A, a)}{\cost(p, a)}\cost(p, c) \leq \cost(O_A, c) + \left\vert \cost(O_A, c) - \frac{\cost(O_A, a)}{\cost(p, a)}\cost(p, c) \right \vert$ finally concludes the proof of Equation~\ref{eq:far}, and hence the lemma.
\end{proof}

\section{Mission Proofs of \cref{sec:dimred}}

\begin{lemma}\label{lem:halfPartitionDimRed}
Let $P \subset \R^d$ be a a multiset with at most  $T$ distinct points. 
Let $\calN$ and $\Pi$ be the set and projection defined in the proof of \cref{lem:partitionDimRed}.

Then,
for any extension $P'$ of $P$ where each point $p$ has coordinate extension $p'$,  and for any subset $C$ of $P$ (with extension $C'$), 
\[\sum_{p \in C} \left\|\zext{p}{p'} - \mu^z_0(C')\right\|^z \leq (1+\eps) \sum_{p \in C}\left\|\zext{\Pi p}{p'}- \mu^z_0(\Pi'(C))\right\|^z\]
where $\Pi'(C) := \left\{\zext{\Pi p}{p'}, p\in C\right\}$ is the projection of $C$ with the same extensions as $C'$.
\end{lemma}
\begin{proof}
Assume that $\costP_0(\Pi'(C))\leq \costP_0(C')$ (otherwise we are already done), and let $S_\Pi$ be a witness set for points $\Pi'(C)$. The assumption on $\costP_0(\Pi'(C))$ ensures that the diameter of $S_\Pi$ is at most $\Delta_{C'}^{1/z}$, where $\Delta_{C'} := \frac{ \costP(C')}{|C'|}$.

Let $S \subset P$ such that $\Pi(S) = S_\Pi$. Since $\Pi$ preserves distances between points of $P$ up to an error 
$(1\pm \eps/z)$, the diameter of $S$ is at most $(1+\eps/z) \Delta_{C'}^{1/z}$: hence, the $\eps' \diam(S)$-cover of $\conv(S)$ is an $\eps'(1+\eps/z) \Delta_{C'}^{1/z}$-cover of $\conv(S)$.

$\Pi$ has two crucial properties. First, since it is a scaled linear projection, it is surjective -- more precisely, the image of $\conv(S)$ is $\conv(\Pi(S))$. 
Second, the image of the cover of $\conv(S)$ is a cover of $\conv(\Pi(S))$; more precisely, we have the following claim: 

\begin{claim}\label{claim:projCover}
    Let $s \in S$ and $c \in \calN$ be its closest point in a $\eps' \diam(S)$-cover of $\conv(S)$. Then, for any $p \in C$, $\|p-c\| \leq  (1+\varepsilon/z) \cdot \left(\|\Pi p - \Pi s\| + \eps/(2z) \cdot \Delta_{C}^{1/z} \right).$
\end{claim}
\begin{proof}
Let $V$ be the orthonormal basis of $S$ added to $\calN$: since $\Pi$ preserves the norm of all vectors in $V$ up to $(1\pm \eps/z)$, and is a linear mapping, its operator norm on the subspace spanned by $S$ is bounded by $(1+\eps/z)$. Therefore:
\begin{align}
    \notag
    \|\Pi s - \Pi c\| &\leq (1+\eps/z) \|s-c\| \\
    \notag
    &\leq (1+\eps/z)  \eps' \diam(S) \leq (1+\eps/z)^2  \eps' \diam(\Pi(S))\\
    \label{eq:projNet}
    &\leq 2\eps' D \Delta_{C'}^{1/z}
\end{align}

Hence, for any point $p \in C$:

\begin{eqnarray*}
& & \|p- c\| \\
(\text{Distortion of Embedding})&\leq & (1+\varepsilon/z)\cdot \|\Pi p- \Pi c\| \\
&\leq & (1+\varepsilon/z)\cdot \left(\|\Pi p-\Pi s\| + \|\Pi s - \Pi c\|\right) \\
(Eq.~\ref{eq:projNet}) &\leq & (1+\varepsilon/z) \cdot \left(\|\Pi p-\Pi s\| +2\eps' D \Delta_{C'}^{1/z} \right) \\
(\text{choice of }\eps') & \leq&  (1+\varepsilon/z) \cdot \left(\|\Pi p - \Pi s\| + \eps/(2z) \cdot \Delta_{C'}^{1/z} \right) \\
\end{eqnarray*}
\end{proof}

Using \cref{lem:weaktri}, we extend the previous inequality as follows: 

\begin{claim}\label{claim:ptoext}
    $\left\|\zext{p}{p'} - \zext{c}{0}\right\|^z \leq (1+3\eps) \cdot (\|\Pi p - \Pi s\|^2 + {p'}^2)^{z/2} + \eps \cdot \Delta_{C'}$
\end{claim}
\begin{proof}
Using \cref{lem:weaktri}, we get $\|p- c\|^2 \leq (1+\eps/z)^3 (\|\Pi p-\Pi s\|^2 + \frac{\eps}{2z} \Delta_{C'}^{2/z})$, and we obtain similarly
    \begin{align*}
\left\|\zext{p}{p'} - \zext{c}{0}\right\|^z &= \left(\|p - c\|^2 + p'^2\right)^{z/2}\\
&\leq \left((1+\varepsilon/z)^3 \cdot \left(\|\Pi p - \Pi s\|^2 +\frac{\eps}{2z}\cdot (\Delta_{C'})^{2/z}\right)  + {p'}^2\right)^{z/2}\\
&\leq (1+\varepsilon/z)^{3z/2} \cdot\left( (\|\Pi p - \Pi s\|^2 + {p'}^2) +\frac{\eps}{2z}\cdot (\Delta_{C'})^{2/z}\right)^{z/2}\\
&\leq (1+\varepsilon/z)^{3z/2}\left( (1+\eps) \cdot (\|\Pi p - \Pi s\|^2 + {p'}^2)^{z/2} + \left(\frac{z+\eps}{\eps}\right)^{z/2-1} \left(\frac{\eps}{2z}\right)^{z/2} \cdot \Delta_{C'}\right)\\
\qquad &\leq (1+3\eps) \cdot (\|\Pi p - \Pi s\|^2 + {p'}^2)^{z/2} + \eps \cdot \Delta_{C'}
\end{align*}
\end{proof}

We can now conclude and show that  $\costP_0(C')\leq (1+O(\eps))\costP_0(\Pi'(C))$. For this, we chose $s$ and $c$ as follows: since $S_\Pi$ is a witness set for points in $\Pi'(C)$, there exists a $(1+\eps)$-approximate center for $\Pi'(C)$ in the convex hull of $S_\Pi$: let $s \in \conv(S)$ such that $\zext{\Pi s}{0}$ is that approximate center -- such an $s$ exists due to the surjectivity of $\Pi$. 
Furthermore, let $c$ be the closest point to $s$ in the cover of $\conv(S)$. With those choices, we get:
\begin{align*}
    \costP_0(C') &\leq \sum_{p \in C} \left\| \zext{p}{p'} - \zext{c}{0}\right\|^z\\
    &\leq \sum_{p\in C}(1+O(\eps))\lpar \|\Pi p - \Pi s\|^2 + p'^2\rpar^{z/2} + \eps \Delta_{C'}\\
    &\leq (1+O(\eps)) \costP_0(\Pi'(C)) + \eps \costP_0(C'),
\end{align*}
and therefore $\costP_0(C')\leq (1+O(\eps))\costP_0(\Pi'(C))$.
\end{proof}

\section{Missing Proofs of \cref{sec:coresetImproved}}
\label{app:proofCoreset}

\begin{lemma}
\mbox{}\label{lem:kepstiny}
It holds that
    \[\max\left(\sum_{p \in I_{tiny,\calS}} \cost(p, \calS), ~
       \sum_{p \in  I_{tiny,\calS} \cap \coreset} 
             \frac{|R_{i,j}|}{\delta}\cost(p,\calS)\right)
     \le \varepsilon \cdot \cost(R_{i,j},\seeded).\]
\end{lemma}
\begin{proof}
By definition of $I_{tiny,\calS}$, $\sum_{p \in I_{tiny,\calS}} \cost(p, \calS) \leq \sum\limits_{p \in I_{tiny,\calS}} \frac{\varepsilon}{2}\cdot  \cost(p, \seeded) \leq \frac{\varepsilon}{2} \cdot \cost(R_{i,j}, \seeded)$. 
Similarly, we have for the other term 
  \begin{align*}
   \sum_{p \in  I_{tiny,\calS} \cap \coreset} \frac{|R_{i,j}|}{\delta} \cdot \cost(p,\calS) 
   & \leq  \sum_{p \in I_{tiny,\calS} \cap \coreset}  \frac{|R_{i,j}|}{\delta} \frac{\varepsilon}{2}\cdot  \cost(p, \seeded) \\
   & \leq  \varepsilon\cdot  \frac{|R_{i,j}|}{\delta} \sum_{p \in  C_i \cap I_{tiny,\calS} \cap \coreset}  \frac{2^z\cost(R_{i,j},\seeded)}{|R_{i,j}|} \\
  & \leq  \varepsilon\cdot  \frac{|I_{tiny, \calS} \cap \coreset|}{\delta}\cost(R_{i,j},\seeded) \leq  \varepsilon\cdot  \cost(G,\seeded)  .
  \end{align*}
  where we used that since each point of $R_{i,j}$ have same cost up to a factor $2^z$, it holds that $\forall p\in R_{i,j}, \cost(p, \seeded) \leq 2^z \frac{\cost(R_{i,j}, \seeded)}{|R_{i,j}|}$.
The last inequality uses that $\coreset$ contains $\delta$ points.
\end{proof}

\begin{lemma}
\label{lem:khuge}
It holds that, for any $R_{i,j}$  and for all solutions $S$ with at least one non-empty huge group $I_{i,j,\ell}$
$$\left\vert\cost(R_{i,j},\calS) - \sum_{p\in \Omega \cap R_{i,j}} \frac{|R_{i,j}|}{\delta} \cdot \cost(p,\calS) \right\vert \leq  3\eps\cdot \cost(R_{i,j},\calS).\qedhere$$
\end{lemma}
\begin{proof}
Fix a ring $R_{i,j}$ and let $I_{i,j,\ell}$ be a huge group. 
First, the weight of  $R_{i,j}$ is preserved in $\coreset$: since the set $\Omega$ has size $\delta$, it holds that 
$$ \sum_{p\in \Omega\cap R_{i,j}} \frac{|R_{i,j}|}{\delta} = |R_{i,j}|$$

Now, let $\calS$ be a solution, and $p\in I_{i,j,\ell}$ with $I_{i,j,\ell}$ being huge. This implies, for any $q \in R_{i,j}$: 
$\cost(p,q) \leq (2\cdot \eps\cdot 2^{j+1})^z\leq 4^z \cdot \eps^z\cdot 2^{(\ell - \log (4z/\eps))z} \leq (\eps/z)^{z}\cdot \cost(p,\calS)$. 
By Lemma~\ref{lem:weaktri}, we have therefore for any point $q\in R_{i,j}$
\begin{eqnarray*}
\cost(p,\calS) &\leq & \left(1+\eps/z \right)^{z-1} \cost(q,\calS) + \left(1+z/\eps\right)^{z-1}\cost(p,q) \\
&\leq & \left(1+\eps \right) \cost(q,\calS) + \varepsilon\cdot \cost(p,\calS) \\
\Rightarrow \cost(q,\calS) &\geq & \frac{1-\varepsilon}{1+\eps}\cost(p,S) \geq  (1-2\eps) \cost(p, \calS)
\end{eqnarray*}

Moreover, by a similar calculation, we can also derive an upper bound of $\cost(q,\calS)\leq \cost(p,\calS)\cdot (1+2\eps)$. Hence, combined with $\sum_{p\in \Omega\cap R_{i,j}} \frac{|R_{i,j}|}{\delta} = |R_{i,j}|$, this is sufficient to approximate $\cost(R_{i,j}, \calS)$. 

Therefore, the cost of $R_{i,j}$ is well approximated for any solution $\calS$ such that there is a non-empty huge group $I_{i,j,\ell}$. 
\end{proof}